%% file: paper.tex
\documentclass[twoside,leqno,twocolumn]{article}

\input{packages}
\input{defines}

\begin{document}

\date{}

\input{titleauth}

\maketitle
\fancyfoot[R]{\scriptsize{Copyright \textcopyright\ 2019 by SIAM\\
Unauthorized reproduction of this article is prohibited}}
\input{abstract}

\input{intro}
\input{motivation}
\input{prelim}
\input{nonuniform}
\input{uniform}
\input{experiments}
\input{related}

\bibliographystyle{abbrv}
\bibliography{bibliography}

\end{document}

%% file: packages.tex
\usepackage{xspace}

\usepackage{amsmath}
\usepackage{amssymb}
\usepackage{csquotes}
\usepackage{mathtools}
\usepackage{cuted}

\usepackage[noadjust]{cite}
\PassOptionsToPackage{hyphens}{url}
\usepackage{hyperref}
\usepackage[capitalize]{cleveref}

\usepackage[usenames, dvipsnames]{color}

\usepackage{subfig}
\usepackage{url}
\usepackage{ltexpprt}

\usepackage{pgfplots}
\usepackage{pgfplotstable}
\usepgfplotslibrary{external}
\tikzsetexternalprefix{pgfplots/}


%% file: defines.tex
\newcommand{\defn}[1]{\emph{\textbf{\boldmath #1}}}

\ifx \shutupcomments \undefined
\newcommand{\alex}[1]{{\textcolor{Magenta}{ [ Alex: #1 ] }}}
\newcommand{\mfc}[1]{{\textcolor{Red}{ [ Martin: #1 ] }}}
\newcommand{\michael}[1]{{\textcolor{Blue}{ [ Michael: #1 ] }}}
\newcommand{\rob}[1]{{\textcolor{purple}{ [ Rob: #1 ] }}}
\newcommand{\mt}[1]{{\textcolor{NavyBlue}{ [ Meng-Tsung: #1 ] }}}
\else
\newcommand{\alex}[1]{{}}
\newcommand{\mfc}[1]{{}}
\newcommand{\michael}[1]{{}}
\newcommand{\rob}[1]{{}}
\newcommand{\mt}[1]{{}}
\fi

\newcommand{\secput}[2]{\section{#2}\label{sec:#1}}
\newcommand{\secref}[1]{\Cref{sec:#1}}     
\newcommand{\subsecput}[2]{\subsection{#2}\label{sec:#1}}

\newcommand{\thmput}[1]{\begin{theorem}\label{thm:#1}}
\newcommand{\lemput}[1]{\begin{lemma}\label{lem:#1}}

\newcommand{\E}[1]{\mathrm{E}\left[#1\right]}
\newcommand{\Var}[1]{\mathrm{Var}\left[#1\right]}
\newcommand{\Prob}[1]{\mathrm{P}\left(#1\right)}
\newcommand{\halfnorm}[1]{\|#1\|_{\frac{1}{2}}}
\newcommand{\halfnormp}{\|\mathbf{p}\|_{\frac{1}{2}}}
\newcommand{\opt}{\mathrm{OPT}}
\newcommand{\ind}[1]{\mathrm{I}\left[#1\right]}
\newcommand{\given}{\:\middle|\:}

\newcommand{\btree}{B-tree}
\newcommand{\btrees}{B-trees}
\newcommand{\betree}{$\textrm{B}^\varepsilon$-tree}
\newcommand{\betrees}{$\textrm{B}^\varepsilon$-trees}

\newcommand{\innodb}{InnoDB}
\newcommand{\mysql}{MySQL}

\newproof{observation}{Observation}

\newcommand{\FB}{\textsc{Fullest Bin}\xspace}

\newcommand{\RB}{\texorpdfstring{\textsc{Random Ball}}{Random Ball}\xspace}
\newcommand{\GG}{\textsc{Golden Gate}\xspace}

\renewcommand{\AE}{\textsc{Aggressive Empty}\xspace}
\newcommand{\p}{\mathbf{p}}
\newcommand{\uni}{\mathbf{u}}

\newcommand{\q}{\mathbf{q}}


%% file: titleauth.tex

\author{
Michael A.~Bender\thanks{
Stony Brook University, 
Stony Brook, NY 11794-2424 USA. 
Email:~\texttt{\{bender,jpchristense\}@cs.stonybrook.edu}.}
\and
Jake Christensen\footnotemark[2]
\and
Alex Conway\thanks{
Rutgers University, Piscataway NJ 08855 USA.  
Email:~\texttt{\{ajc179,farach\}@cs.rutgers.edu}.}
\and
Martin Farach-Colton\footnotemark[3]
\and
Rob Johnson\thanks{
VMware Research,
Creekside F,
3425 Hillview Ave,
Palo Alto, CA 94304 USA.
Email:~\texttt{robj@vmware.com}.}
\and
Meng-Tsung Tsai\thanks{
National Chiao Tung University, Hsinchu, Taiwan.
Email:~\texttt{mtsai@cs.nctu.edu.tw}.}}

\title{Optimal Ball Recycling\thanks{
This research was supported in part by NSF grants 
CCF 1439084, 
CCF 1637458, 
CNS 1408695, 
CNS 1755615, 
CNS 1763680, 
CNS 1408782, 
IIS 1247726, 
IIS 1251137, 
CCF 1617618, 
CCF 1725543, 
CCF-BSF 1716252, 
and IIS 1541613, 
NIH grant 1U01CA198952-01, 
MOST grant 107-2218-E-009-026-MY3, 
and by EMC, Inc, NetAPP, Inc, and Sandia National Labs.
}}


%% file: abstract.tex

\begin{abstract}

	Balls-and-bins games have been a successful tool for modeling load
	balancing problems.  In this paper, we study a new scenario, which we call
	the \defn{ball-recycling game}, defined as follows:
	\begin{displayquote}
		Throw $m$ balls into $n$ bins i.i.d.\ according to a given probability
		distribution $\p$.  Then, at each time step, pick a non-empty bin and
		\defn{recycle} its balls: take the balls from the selected bin and
		re-throw them according to $\p$.  
	\end{displayquote}
	This balls-and-bins game closely models memory-access heuristics in
	databases.  The goal is to have a bin-picking method that maximizes the
	\defn{recycling rate}, defined to be the expected number of balls recycled
	per step in the stationary distribution.

	We study two natural strategies for ball recycling: \FB, which greedily
	picks the bin with the maximum number of balls, and \RB, which picks a ball
	at random and recycles its bin.  We show that for general $\p$, \RB is
	$\Theta(1)$-optimal, whereas \FB can be pessimal. However, when $\p =
	\uni$, the uniform distribution, \FB is optimal to within an additive
	constant.

\end{abstract}


%% file: intro.tex

\section{Introduction}

Balls-and-bins games have been a successful tool for modeling load balancing
problems~\cite{Mitzenmacher96,AdlerChMi98,AdlerBeSc98,AzarBrKa94,ColeFrMa98,ColeMaMe98,CzumajSt97,Mitzenmacher96a,Mitzenmacher98,MitzenmacherRiSi00,Voecking03,BringmannSaSt16,Park17,WestZaWa16,farach-Colton:2009,conway_et_al:LIPIcs:2018:9043}.
For example, they can be used to study the average and worst-case occupancy of
buckets in a hash table~\cite{BenjaminiMa12}, the worst-case load on nodes in a
distributed cluster~\cite{PetrovaOlMa10,BerenbrinkFrKl16} and even the amount
of time customers wait in line at the grocery store~\cite{Mitzenmacher98}.  In
all these load-balancing problems, balls-and-bins games are used to study how
to distribute load evenly across the resource being allocated.

In this paper we study a new scenario, which we refer to as the
\defn{ball-recycling game}, defined as follows:
\begin{displayquote}
	Throw $m$ balls into $n$ bins i.i.d.\ according to a given probability
	distribution $\p$.  Then, at each time step, pick a non-empty bin and
	\defn{recycle} its balls: take the balls from the selected bin and re-throw
	them according to $\p$.
\end{displayquote}
We call a bin-picking method a \defn{recycling strategy} and define its
\defn{recycling rate} to be the expected number of balls recycled in the
stationary distribution (when it exists).

The ball-recycling game models \defn{insertion buffers} and \defn{update
buffers}, which are widely used to speed up insertions in databases by batching
updates to blocks on disk. The recycling rate of a recycling strategy
corresponds to the speed-up obtained by an insertion buffer, so the goal
studied in this paper is how to maximize the recycling rate. This relationship
is described in \Cref{sec:motivation}, and the experiments in
\Cref{sec:experiments} demonstrate that it holds in practice.

In this paper, we present results for ball recycling for both general $\p$ and
for the special case of uniform $\p$, which we denote by $\uni$.  As we explain
in \Cref{sec:motivation}, these distributions correspond to update and
insertion buffers, respectively.

We focus on three natural recycling strategies:
\begin{itemize}
\item \FB: A greedy strategy that recycles the bin with the most balls.
\item \RB: A strategy that picks a ball uniformly at random and recycles its
	bin.
\item \GG: A strategy that picks the bins in round-robin fashion; after a bin
	is picked, the next bin picked is its non-empty successor.
\end{itemize}
Let $\halfnormp = (\sum \sqrt{\p_i})^2$  be the half quasi-norm of $\p$.  We
achieve the following result for general $\p$.

\begin{theorem}[\Cref{sec:nonuniform}]\label{thm:random-opt}
	Consider a ball-recycling game with $m$ balls and $n$ bins, where the balls
	are distributed into the bins i.i.d.\ according to distribution $\p$. Then
	\RB{} is $\Theta(1)$-optimal.

        It achieves recycling rate
	$\mathcal{R}^\textrm{RB}$:
	\begin{enumerate}
		\item If $m \geq n$,
			\[\mathcal{R}^\textrm{RB} = \Theta\left(\frac{m}{\halfnormp}\right).\]
		\item If $m < n$, let $L$ be the $m$ lowest-weight bins, $q = \sum_{\ell\in
			L} p_\ell$, and $\mathcal{R}_L^\textrm{RB}$ be the recycle rate of
			\RB restricted to $L$. Then,
			\[\mathcal{R}^\textrm{RB} =
			\Theta\left(\min\left(\mathcal{R}_L^\textrm{RB}, 1/q\right)\right).\]
	\end{enumerate}
\end{theorem}

In order to establish this result, we first show that no recycling strategy can
achieve a higher recycling rate than $(2m+n)/\halfnormp$.  This directly
establishes optimality when $m = \Omega(n)$. For $m=o(n)$, we show that \RB
performs as well as another strategy, \AE, which takes an optimal strategy on a
subset of high-weight bins and turns it into an optimal strategy on all the
bins.

Interestingly, the greedy strategy \FB is not generally optimal, and in
particular:
\begin{observation}
There are distributions for which \FB is pessimal, that is, it
recycles at most 2 balls per round whereas OPT recycles almost $m$
balls per round.
\end{observation}
For example, consider the \defn{skyscraper distribution}, where $\p_0 =
1-1/n+1/n^2$ and $\p_i = 1/n^2$, for $0<i\leq n-1$.  Suppose that $m=\sqrt{n}$.
Then \FB will pick bin $0$ every time until it has at most one ball, at which
point it will pick another bin, which will almost certainly have $1$ ball in
it.  Thus, the recycling rate of \FB drops below $2$.  Suppose, instead, that
we recycle the least-full non-empty bin.  In this case, every approximately
$\sqrt{n}$ rounds, a ball lands in a low-probability bin and is promptly
returned to bin $0$. Thus, the recycling rate of this strategy is nearly $m$.
Thus, \FB is pessimal for this distribution.

However, the uniform distribution is of particular importance to insertion
buffers for databases based on \btrees{}. This is because (arbitrary) random
\btree{} insertions are nearly uniformly distributed across the leaves of the
\btree{}, as we show in \Cref{sec:motivation}. On the uniform distribution, \FB
and \GG are optimal even up to lower order terms:

\begin{theorem}[\Cref{sec:uniform}]\label{thm:fullestbin}
  \FB and \GG are optimal to within an additive constant for the
  ball-recycling game with distribution $\uni$ 
  for any $n$ and $m$.  They each achieve a
  recycling rate of at least $2m/(n+1)$, whereas no recycling strategy can
  achieve a recycling rate greater than $2m/n + 1$.
\end{theorem}

In this case, \RB is only optimal to within a multiplicative constant
in the following range:

\begin{theorem}[\Cref{sec:uniform}]\label{thm:rbuniform}
	On the uniform distribution $\uni$, for sufficiently large
        $m$, \RB is at least $(1/2+1/(2^3
	3^4))$-optimal and at most $(1-3/1000)$-optimal.
\end{theorem}

Thus, we establish some surprising results: that \FB can perform poorly for
arbitrary $\p$ but is optimal for $\uni$, up to lower-order terms; and that \RB
is asymptotically optimal for any $\p$ and in particular is more than
$1/2$-optimal but not quite optimal for $\uni$.

In \Cref{sec:experiments}, we present experimental results showing that our
analytical results for the ball recycling problem closely matches performance
results in real databases.  We describe the recycling strategies of several
commercial and open-source database systems. In particular we focus on InnoDB,
a \btree{} that uses a variant of \RB.

Our results suggest that \FB or \GG would be a better choice than \RB for
InnoDB. In particular, \GG requires almost no additional bookkeeping, and can
be implemented in InnoDB with a change of only a few lines of code.
With this implementation, we measured a 30\% improvement in its
insertion-buffer flushing rate, which is in line with our theoretical results.

We conclude that ball recycling is a natural hitherto unexplored balls-and-bins
game that closely models a widely deployed method for improving the performance
of databases. Moreover, this is the first application (to our knowledge) of a
balls-and-bins game to the throughput of a system. This is in contrast to past
balls-and-bins analyses, which modeled load balancing and latency.


%% file: motivation.tex

\section{Ball Recycling and Insertion/Update Buffers}\label{sec:motivation}

Ball recycling models insertion and update buffers, which are widely used in
modern databases~\cite{Azure16, IBM17, Xiang12, Informix, Callaghan11, NuDB16,
Oracle17, SAP17, Vertica17,CanimLaMi10,BenderFaJo12}. These implementations are
discussed in more detail in \Cref{sec:experiments}.

For a key-value store, such as a database, an \defn{insertion buffer} is a
cache of recently inserted items.  When the insertion buffer fills, the
database selects a disk block and all the cached items going to that block are
evicted in bulk.  If $k$ elements are flushed in bulk, then there is a speedup
of $k$, compared to writing the elements to the destination block as soon as
they arrive. After evicting $k$ items, there is room for $k$ new elements in
the buffer.  An \defn{update buffer} caches changes to existing key-value pairs
but is otherwise like an insertion buffer.  Although these types of buffers
seem quite similar, we show that they have important differences in how they
are modeled as a ball-recycling game.

The mapping to ball recycling is direct: disk blocks are bins and elements in
the insertion/update buffer are balls.  The probability distribution $\p$ is
based on the distribution of items inserted  or updated.  Evicting all the
items going to a disk block corresponds to emptying the bin associated with
that disk block.  After an eviction of $k$ items, we have room for $k$ new
insertions/updates, i.e., we have $k$ new balls to throw.  The policy for
selecting the target disk block of an eviction corresponds to the policy for
selecting a bin to recycle, and the speedup induced by an eviction policy is
its recycling rate.  

For \btrees{}, insertion buffers and update buffers differ in an important way:
updates do not change the structure of leaves of a \btree{}. In contrast,
insertions can change the range of keys associated with a leaf (due to leaf
splits), which yields the following result:

\begin{lemma}\label{lem:uniform-leaves}
	If $N$ keys are inserted into a \btree{} i.i.d.\ according to some key
	distribution $\q$, then provided $B = \Omega(\log{N})$, the maximum
	probability that a leaf has of receiving the next insertion is $O(B/N)$
	with probability 1.  Thus the corresponding recycling game is
	asymptotically \defn{almost uniform}: no bin has probability more than a
	constant multiple of $\frac{1}{n}$.
\end{lemma}

We prove the uniformity bound as follows.  Let $F(\kappa)$ be the cumulative
density function (CDF) of $\q$, which is the probability that an item sampled
from $\q$ is less than $\kappa$.  If $\kappa$ is distributed according to $\q$,
then $F(\kappa)$ will be uniformly distributed on $[0,1]$.

If $n$ points are sampled from $[0,1]$ uniformly and sorted so that $x_1
\leq x_2 \leq \cdots \leq x_n$, then $\max{x_{i+B}-x_i}$ is known as the
\defn{maximal $B$-spacing}. It follows that:

\begin{lemma}\label{lem:leaf-b-spacing}
	Having inserted $n$ keys into a \btree{} i.i.d.\ according to a distribution
	$\q$, the maximum probability that any leaf has of receiving the next
	insertion is less than the maximum $B$-spacing of the CDFs of those points.
\end{lemma}

It is known that:

\begin{lemma}[\cite{DeheuvelsDe84}]
	If $B = \Theta(\log{n})$, then the maximum $B$-spacing of $n$ points
	distributed uniformly on the unit interval is $\Theta(B/n)$ with
	probability 1.
\end{lemma}

For $B = \omega(\log{n})$, we can subdivide $B$ into intervals of $\log{n}$
points, each of which will satisfy the lemma. Adding together the resulting
bounds, we have that the maximal $B$-spacing of $n$ points is $O(B/n)$ with
probability 1. Together with \Cref{lem:leaf-b-spacing}, this implies
\Cref{lem:uniform-leaves}.

We also note that an (almost-uniform) ball-recycling game is an imperfect model
for an insertion buffer, because the ball-recycling game has a fixed number of
bins, whereas in the insertion buffer, the number of disk blocks will increase.
Furthermore, insertions may not be independently distributed.

Finally, the implementation of these strategies is a point of departure between
insertion/update buffers and ball recycling.  In ball recycling, it is obvious
which bin each ball is in.  In insertion/update buffers for a \btree{},
elements have a key, but we don't necessarily know what the buckets are, since
the mapping from keys to buckets depends on the pivots used to define the
\btree{} leaves.  \FB needs to know what the buckets are, whereas \RB and \GG
do not.  For \RB this is because the key of the randomly selected item can be
used to fetch its target \btree{} leaf, after which we know the max and min
keys in that leaf, and \GG can be implemented by remembering the upper bound of
the last leaf to which we flushed and then flushing the item with the successor
of that key, along with all the other keys going to that leaf.  None of these
strategies require knowledge of $\p$. 

Our results on general $\p$ have an interesting implication for \betrees{},
which are known to be asymptotically optimal for insertions, in the worst case.
\betrees{} can also handle updates by propagating messages to the leaves.  For
some update distributions, flushing according to \RB can achieve an update rate
that is $B^\varepsilon$ faster than \FB.  We expect to try \RB flushing in our
\betree{}-based file system,
B$\varepsilon$trFS~\cite{JannenYuZh15a,YuanZhJa16,YuanZhJa17,ConwayBa17,JannenYuZh15,esmet2012tokufs}.


%% file: prelim.tex

\secput{prelim}{Ball Recycling and Markov Theory}

We begin our analysis of ball-recycing games with some preliminary results.  In
particular, we show that all finite-state ball-recycling strategies have
stationary distributions.

\subsecput{mdps}{Ball-recycling games are Markov decision processes.}

This section makes use of the standard theory of Markov chains and Markov
decision processes; for an introduction see {\it e.g.}
Kallenberg~\cite{Kallenberg16}.

In a ball-recycling game, we represent the configuration of the balls as a
vector $X=(X_i)$ of length $n$, where $X_i$ is the number of balls in the $i$th
bin.  Since the number of balls is finite, there are only a finite number of bin
configurations.

A recycling strategy $A$ takes as input the current bin configuration $X$
together with an internal state $S$, and selects a non-empty bin to recycle;
the next state is obtained by removing all the balls from the selected bin and
re-throwing them according to $\p$. The bin selection may be randomized.  We
write $A^iX$ for the state obtained after $i$ rounds of recycling using
strategy $A$.  In each round, the recycling algorithm earns a reward equal to
the number of balls recycled in that round.

In this way, the ball-recycling game is a Markov decision process, and we are
interested in policies that maximize the expected average recycling rate,
defined for a policy $A$ as
\begin{equation*}
	\mathcal{R}^A = \lim_{T\rightarrow\infty} \frac{1}{T}\sum_{t=0}^{T} R(A^tX_0).
\end{equation*}

Note that Markov decision processes are very general.  For example, in a Markov
decision process the policy may vary over time, and may even take the entire
history of the process and its own past decisions into account when deciding on
its next action.  Thus, for some strategies $A$, the limit $\mathcal{R}^A$ may
not exist.  In the literature of Markov decision processes, this is often
handled by taking the $\limsup$ instead of the limit. However, for any Markov
decision process, any strategy that maximizes the limit also maximizes the
$\limsup$~\cite{Kallenberg16}.  Therefore, for simplicity, we will focus only
on strategies for which the limit is well-defined, and the results will
generally hold for the $\limsup$ of arbitrary strategies.

A Markov decision process policy is \defn{deterministic} if it decides on its
next action based solely on the current state, \textit{i.e.}\ without looking at
history, the number of time steps that have passed or by flipping random
coins.  A deterministic policy can be represented as a simple table mapping
each state to a single action to be taken whenever the system is in that state.

The specific strategies we analyze are \defn{finite-state} strategies, where
the internal state has only finitely many configurations.  When we prove our
lower bounds, we will further restrict ourselves to \defn{stateless}
strategies, where there is a unique internal state. In order to do so, we make
use of the following lemma.

\begin{lemma}[\cite{Kallenberg16}]\label{lem:deterministic-opt}
	There exists a stateless deterministic recycling strategy $\opt$ that
	achieves the optimal expected average recycling rate.
\end{lemma}
\begin{proof}
	This follows from Kallenberg's Corollary 5.4.
\end{proof}

\subsecput{markov}{Stationary distributions of recycling strategies}

A ball-recycling game and a recycling strategy together define a Markov
process on the state space; the space of pairs comprising a balls-and-bins
configuration and an internal state. If the strategy is stateless, this is a
Markov process on the balls and bins space.

There are a finite number of balls-and-bins configurations. Therefore, a
ball-recycling game with a finite-state recycling strategy defines a finite
Markov process, and so has at least one stationary distribution.

We now show that stateless recycling strategies result in Markov processes with
unique stationary distributions.  The following lemma shows that, when we look
only at the bin configurations, recycling games have properties analogous to
irreducibility and aperiodicity in Markov chains. 

\begin{lemma}\label{lem:ergodic}
	For any ball-recycling game with $m$ balls and $n$ bins there is an
	$\epsilon>0$ such that, for all bin configurations $X$ and $Y$, and for all
	recycling strategies, the probability that $X$ reaches $Y$ within
	$\min(m,n)$ steps is at least $\epsilon$.  Furthermore, every bin
	configuration can transition to itself in one time step.
\end{lemma}
\begin{proof}
	We just need to show a sequence of outcomes for ball tosses that transform
	$X$ into $Y$, no matter which bins the recycling strategy chooses to empty.
	So, imagine that, at each step, all the recycling balls land in occupied
	bins, so that at each step, the number of occupied bins goes down by 1.
	After at most $\min(m,n)-1$ steps, all the balls will be in a single bin.
	On the next round, the recycling strategy must choose that bin, causing all
	the balls to be rethrown.  There is some non-zero probability that they
	land in configuration $Y$.

	For the second observation, simply note that all the recycled balls may
	happen to land in the bin from whence they came.
\end{proof}

\begin{lemma} \label{lem:stationary}
	\leavevmode
	\begin{enumerate}
	\item All ball-recycling games using stateless recycling strategies have
		unique stationary distributions which are equal to their limiting
		distributions.
  \end{enumerate}
\end{lemma}
\begin{proof}
	Having fixed a stateless recycling strategy, the ball-recycling game is a
	Markov process on the balls-and-bins configuration. By
	\Cref{lem:ergodic}, this process is irreducible and aperiodic, and so
	has a unique stationary distribution equal to its limiting distribution.
\end{proof}

Together with \Cref{lem:deterministic-opt}, we have:
\begin{corollary} \label{cor:opt-unique}
	For any ball-recycling game, there exists an optimal strategy with a unique
	stationary distribution.
\end{corollary}

We now show that two of the three main recycling strategies studied in this
paper yield unique stationary distributions.  The strategies are:
\begin{itemize}
\item \FB: selects the bin that has the most balls;
\item \RB: selects a ball uniformly at random and recycles whichever bin it is in;
\item \GG: selects the bins in a round robin sequence.
\end{itemize}

\FB is a deterministic strategy, \RB is a stateless strategy and \GG is a
finite-state strategy. By \Cref{lem:stationary} we have:

\begin{lemma} \label{lem:fb-gg-rb-ergodic}
  \FB and \RB have unique stationary distributions.
\end{lemma}


%% file: nonuniform.tex

\secput{nonuniform}{Random Ball is Optimal}

In this section we prove \Cref{thm:random-opt}, showing that \RB{} is
$\Theta(1)$-optimal. 

\subsecput{nonuni-outline}{Outline of Proof}
We prove \Cref{thm:random-opt} in the following steps:
\begin{enumerate}
	\item No recycling strategy has a recycling rate exceeding
		$(2m+n-1)/\halfnorm{\p}$. (\secref{upper-statement})
	\item \RB{} matches that bound when $m \geq n$, leading to case (1) of
		\Cref{thm:random-opt}. (\secref{randomballwithmanyballs})
	\item There is an $\Theta(1)$-optimal strategy, \AE{}, when $m < n$. The
		recycling rate of \AE matches case (2) of \Cref{thm:random-opt}.
		(\secref{randomballaggroempty})
	\item By comparison to \AE{}, \RB{} is $\Theta(1)$-optimal when $m < n$.
		(\secref{rboptproof})
\end{enumerate}

\subsecput{upper-statement}{The Upper Bound}
We begin by proving an important lemma that will be used throughout, which we
refer to as the \defn{flow equation}. Then we proceed to prove the upper bound.

Let $A$ be a stateless strategy with stationary distribution $\chi^{A,\p}$. Let
$\phi_i$ be the event that $A$ picks bin $i$ to recycle. Let $\mathcal{R}_i^A =
E[R^A(\chi^{A,\p}) | \phi_i]$ be the number of balls recycled given that the
strategy picks bin $i$, and $f_i = P(\phi_j)$, the probability of picking the
bin $i$ in the stationary distribution. We note that $\mathcal{R}_i^A$ and
$f_i$ could alternatively be defined as limits of repeated applications of $A$
to any given starting state.

For a given bin $i$, we can analyze the ``flow'' of balls into and out of $i$.
When $k$ balls are thrown, $\p_ik$ of them are expected to land in $i$. For a
ball to leave $i$, $i$ must first be picked to be emptied by $A$, at which
point every ball in $i$ will be evicted. In the stationary distribution, the
net flow must be zero.  We can generalize to any set of bins and get:

\begin{lemma} \label{lem:gen-flow-equation}
	Let $A$ be a statelss strategy for a ball-recycling game with $n$ bins with
	probabilities $\p$. If $L$ is a subset of the bins, $\p_L = \sum_{\ell \in
	L} \p_\ell$, $f_L = \sum_{\ell \in L} f_\ell$ and $\mathcal{R}_L^A$ the
	conditional recycle rate given $A$ picks a bin in $L$, then
	\begin{equation}
		\p_L\mathcal{R}^A = f_L\mathcal{R}_L^A.
	\end{equation}
\end{lemma}

We will mostly use the following special case of the
\Cref{lem:gen-flow-equation}: 

\begin{lemma}[The Flow Equation] \label{lem:flow-equation}
	Let $A$ be a statless recycling strategy for a ball-recycling game with $n$
	bins with probabilities $\p$.  Then, for all $0\leq i < n$,
	\begin{equation}
		\p_i\mathcal{R}^A = f_i\mathcal{R}_i^A.
	\end{equation}
\end{lemma}

We now describe the main upper bound on the recycling rate of any recycling
strategy. In order to understand the intuition behind \Cref{lem:freqbound},
consider a given bin $i$. Intuitively, it makes sense to think that for a
reasonable recycling strategy the recycling rate of the other bins in the
system will go down as the number of balls $X_i$ in bin $i$ grows.  After all,
the $X_i$ balls in bin $i$ aren't available for recycling, until bin $i$ is
selected. If we assume this intuition as fact for the moment, this suggests
that the expected number of balls in bin $i$ should be greater than half the
recycling rate of bin $i$, perhaps excluding the last ball to land in the bin.

By the Flow Equation, this would suggest that
\[ \E{X_i} \geq \frac{1}{2}(\mathcal{R}_i^A - 1) =
\frac{1}{2}\left(\frac{\p_i}{f_i}\mathcal{R}^A - 1\right), \]
so that after summing over $i$, we obtain \Cref{lem:freqbound}.

However, the following strategy does not satisfy this assumption: for a given
bin $i$, have the strategy just pick the least full non-empty bins until $i$
has a few balls, then pick the fullest ones, then pick $i$ and repeat. Showing
that better strategies do not do this is non-trivial, and we prove
\Cref{lem:freqbound} by different means.

\begin{lemma}\label{lem:freqbound}
	Consider a ball-recycling game with $m$ balls, $n$ bins and distribution
	$\p$. If $A$ is a stateless strategy that picks bin $i$ with frequency
	$f_i$, then its recycle rate is bounded by
	\begin{equation}
		\mathcal{R}^A \leq \frac{2m+n-1}{\sum_j\frac{\p_j}{f_j}}.
	\end{equation}
\end{lemma}

Given a strategy $A$, the idea of the proof is to use the invariance of the
statistic
\begin{equation}
	Z(X) = \sum_{j = 1}^n \frac{X_j^2}{\p_j},
\end{equation}
under the action of $A$ on its stationary distribution. The application of $A$
to $Z$ together with the flow equation creates a factor of $\sum_j
\mathcal{R}_j$, which when solved for proves the bound. First, we begin with
some foundational lemmas, and then proceed to prove the main results.

\begin{lemma}
	Suppose $k$ balls are thrown into $n$ bins i.i.d.\ according to distribution
	$\mathbf{p}$. Let $B(j,k)$ be the binomial random variable denoting how
	many balls land in the $j$th bin.  The following hold:
	\begin{gather}
		\E{B(j,k)} = \p_jk \\
		\E{(B(j,k))^2} = \p_j(1-\p_j)k + \p_j^2k^2
	\end{gather}
\end{lemma}

\begin{proof}
	$B(j,k)$ is a binomial random variable with parameters $\p_j$ and $k$.
\end{proof}

Next, given a state $X$, we compute the effect of recycling the $\ell$th bin on
the $j$th component of $Z$. Note that if $A$ recycles bin $\ell$ of state $X$,
then $X_\ell = R^A(X)$.
\begin{lemma}
	In a ball-recycling game with $m$ balls, $n$ bins and probability
	distribution $\mathbf{p}$, if a strategy $A$ recycles bin $\ell$ in state
	$X$, then for $j \neq \ell$,
	\begin{multline}
		\E{(AX)_j^2}=X_j^2 + 2X_j\p_jR^A(X) \\
		+ \p_j(1-\p_j)R^A(X) + \p_j^2R^A(X)^2,
	\end{multline}
	where $R^A(X) = X_\ell$ is the number of balls recycled.
\end{lemma}

\begin{proof}
	\begin{align*}
		\E{(AX)_j^2} &= \E{(X_j+B(j,R^A(X)))^2} \\[1ex]
					 &= \begin{multlined}[t] X_j^2 + 2X_j\E{B(j,R^A(X))}\\
						 +\E{B(j,R^A(X))^2} \end{multlined} \\[1ex]
					 &= \begin{multlined}[t] X_j^2 + 2X_j\p_jR^A(X)\\+\p_j(1-\p_j)R^A(X)
						 + \p_j^2\left(R^A(X)\right)^2 \end{multlined}
	\end{align*}
\end{proof}

We now can use this result to compute the result of applying $A$ to $Z$. 
\begin{lemma}
	In a ball-recycling game with $m$ balls, $n$ bins and probability
	distribution $\mathbf{p}$, if a strategy $A$ recycles bin $\ell$ in state
	$X$, then
	\begin{multline}
		\E{Z(AX)}=Z(X) - \left(1+\frac{1}{\p_\ell}\right)\left(R^A(X)\right)^2 \\
		+ (2m+n-1)R^A(X)
	\end{multline}
\end{lemma}

\begin{proof}
	\begin{align*}
		\E{Z(AX)} 
		&= \sum_{j=1}^n \E{(AX)_j^2/\p_j} \\
		&= \frac{\E{(AX)_\ell^2}}{\p_\ell} + \sum_{j \neq \ell} \frac{\E{(AX)_j^2}}{\p_j} \\
		&= \begin{multlined}[t][6.25cm](1-\p_\ell)R^A(X) + \p_\ell\left(R^A(X)\right)^2 \\[0.75ex]
			\vspace*{-1.5ex}\shoveleft[0.25cm]{+ \sum_{j \neq \ell} \left( X_j^2/\p_j + 2X_jR^A(X)} \vphantom{\left(R^A(X)\right)^2}\right.\\
			\left.+ (1-\p_j)R^A(X) + \p_j\left(R^A(X)\right)^2 \right) \end{multlined} \\
		&= Z(X) - \left(2 + \frac{1}{\p_\ell}\right)\left(R^A(X)\right)^2 + 2mR^A(X) \\
		&\qquad+ \sum_j \left((1-\p_j)R^A(X) + \p_j\left(R^A(X)\right)^2\right) \\
		&= \begin{multlined}[t]Z(X) - \left(1+\frac{1}{\p_\ell}\right)\left(R^A(X)\right)^2 \\
			+ (2m+n-1)R^A(X) \end{multlined}
	\end{align*}
\end{proof}

Now, we can prove \Cref{lem:freqbound}.

\begin{proof}[Proof of \Cref{lem:freqbound}]
	Let $\chi^A$ be the stationary distribution relative to $A$. Let $\phi_j$
	be the event that $A$ recycles the $j$th bin of $\chi^A$, $R_j^A$ the
	random variable of how many balls are recycled given the $j$th bin is
	chosen by $A$, $\mathcal{R}_j^A=\E{R_j^A}$ and $f_j$ the probability that
	$A$ recycles that bin.  Because $\chi^A = A\chi^A$ by definition, we must
	have $\E{Z(A\chi^A)} = \E{Z(\chi^A)}$. Therefore:

	\begin{align*}
	\E{Z(\chi^A)} &= \E{Z(A\chi^A)}\\
	  			  &= \sum_j f_j \E{Z(A\chi^A)|\phi_j} \\
				  &= \begin{multlined}[t][6.5cm]\sum_j f_j \left(\E{Z(\chi^A)|\phi_j}
					  \vphantom{\left(1+\frac{1}{\p_j}\right)} \right. \\
					  - \left(1+\frac{1}{\p_j}\right)\E{\left(R_j^A\right)^2} \\
					  \left. \vphantom{\left(1+\frac{1}{\p_j}\right)} + (2m+n-1)\mathcal{R}_j^A\right) \end{multlined} \\
				  &\leq \begin{multlined}[t][6.5cm]\E{Z(\chi^A)} + (2m+n-1)\mathcal{R}^A \\
					  - \sum_j f_j\left(1 +\frac{1}{\p_j}\right)\left(\mathcal{R}_j^A\right)^2 \end{multlined} \\
				  &= \begin{multlined}[t][6.5cm]\E{Z(\chi^A)} + (2m+n-1)\mathcal{R}^A \\
					  - \sum_j \frac{1}{f_j}\left(\p_j^2 +\p_j\right)\left(\mathcal{R}^A\right)^2, \end{multlined}
	\end{align*}
	where the inequality is due to the Cauchy-Schwartz Inequality and the last
	line is because of the Flow Equation.
	
	Thus we have:
	\begin{align*}
		\mathcal{R}^A &\leq \frac{2m+n-1}{\sum_j \frac{1}{f_j}\left(\p_j^2 +\p_j\right)} \\
					&\leq \frac{2m+n-1}{\sum_j \frac{\p_j}{f_j}}
	\end{align*}
\end{proof}

\Cref{lem:freqbound} applies to the optimal deterministic strategy $\opt$
promised by \Cref{cor:opt-unique}, and we know that $\mathcal{R}^A\leq
\mathcal{R}^\opt$ for \textit{any} recycling strategy $A$.  Thus, by maximizing
the RHS of \Cref{lem:freqbound}, we can get an upper bound on the recycling
rate of any recycling strategy.

\begin{lemma}\label{lem:upperbound}
	Consider a ball-recycling game with $m$ balls, $n$ bins and distribution
	$\p$. For any recycling strategy $A$, 
	\[\mathcal{R}^A \leq \frac{2m+n-1}{\halfnorm{\p}}.\]
\end{lemma}

\begin{proof}
	This follows immediately from the Cauchy-Schwartz Inequality.
\end{proof}

\subsecput{randomballwithmanyballs}{\RB with \texorpdfstring{$m \geq n$}{m >= n}}

We show the following lower bound, which with \Cref{lem:upperbound}, shows
optimality when $m = \Omega(n)$.

\begin{lemma}\label{lem:random-ball-lower}
	\RB recycles at least $\frac{m}{\halfnormp}$ balls per round in
	expectation. 
\end{lemma}

\begin{proof}
	Let $\chi^\mathrm{RB}=(\chi_i^\mathrm{RB})$ be the random variable of the
	number of balls in each bin in the stationary distribution of \RB.  \RB
	recycles bin $i$ with probability $\frac{\chi_i^\mathrm{RB}}{m}$, and
	therefore the expected number of balls recycled from bin $i$ per round is
	$\frac{\E{\left(\chi_i^\mathrm{RB}\right)^2}}{m}$. The number of balls that
	land in bin $i$ per round is
	$\p_i\sum_{j=1}^n\frac{\E{\left(\chi_j^\mathrm{RB}\right)^2}}{m}$.  Since
	$X$ is distributed stationarily, we must have

	\begin{equation*}
		\p_i\sum_{j=1}^n\frac{\E{\left(\chi_j^\mathrm{RB}\right)^2}}{m}
	= \frac{\E{\left(\chi_i^\mathrm{RB}\right)^2}}{m} \geq \frac{\E{\chi_i^\mathrm{RB}}^2}{m},
	\end{equation*}

	using Jensen's Inequality. Clearing denominators, taking square roots and
	summing across $i$, we have
	\begin{align*}
		\left(\sum_{j=1}^n \E{\left(\chi_j^\mathrm{RB}\right)^2}\right)^{\hspace*{-0.5ex}\frac{1}{2}} \hspace*{-0.5ex}\sum_{i=1}^n \p_i^\frac{1}{2}
		&= \sum_{i=1}^n \left( \p_i\sum_{j=1}^n {\E{\left(\chi_j^\mathrm{RB}\right)^2}} \right )^{\hspace*{-0.5ex}\frac{1}{2}} \\
		&\geq \sum_{i=1}^n \E{\chi_i^\mathrm{RB}} = m.
	\end{align*}
	Therefore the expected recycle rate is
	\[ \sum_{j=1}^n \frac{\E{\left(\chi_j^\mathrm{RB}\right)^2}}{m}
	\geq \frac{m}{\left(\sum_{i=1}^n \sqrt{\p_i}\right)^2}
	= \frac{m}{\halfnormp}. \]
\end{proof}

\begin{corollary}\label{cor:rboptmanyballs}
	Consider a ball-recycling game with $m$ balls and $n$ bins.  If $m =
	\Omega(n)$, then \RB is asymptotically optimal among recycling strategies. 
\end{corollary}

\subsecput{randomballaggroempty}{\AE is Optimal}

In this section, we investigate \AE strategies, which aggressively recycle
balls outside a given subset of bins.  An \AE strategy runs one strategy on a
fixed subset of bins, but always chooses to recycle a bin outside of this set
if there exists one which has any balls.  Specifically, we show that a
$\Theta(1)$-optimal strategy on a particular $O(m)$ subset of the bins can be
extended to an $\Theta(1)$-optimal strategy on the full ball-recycling game by
aggressively emptying the
rest.

Consider a ball-recycling game with $m$ balls, $n$ bins, and ball distribution
$\p$. Let $L$ be some subset of bins and $S$ be a strategy on the \defn{induced
ball-recycling game} of $L$, which is the ball-recycling game with $m$ balls,
$|L|$ bins, and ball distribution $\q$, where 
\begin{equation*}
	\q_i = \frac{\p_i}{\sum_{\ell \in L} \p_\ell}.
\end{equation*}
Therefore, $\q$ is $\p$'s conditional probability distribution on $L$. We
define $L,S$-\AE to be the strategy which empties the lowest weight non-empty
bin in the complement of $L$ if one exists and otherwise performs $S$ on $L$.
Note that all the balls will be in $L$ whenever $S$ is performed, so this is
well-defined.

We begin by showing that there exists an $L$ and $S$ such that $|L|=O(m)$, $L$
contains all bins with weight at least $\frac{1}{m}$, and $L,S$-\AE is
asymptotically optimal. Note that when $m = \Omega(n)$, this is trivial,
because we can take $L$ to be all the bins and $S$ to be a $\Theta(1)$-optimal
strategy; however, this section provides stronger bounds when $m = o(n)$.
Intuitively, the idea is that very low weight bins won't be able to effectively
accumulate balls, so strategies do better to recover any balls in them than to
wait for more balls to land there.

\begin{lemma}\label{lem:aggroopt}
	There exists an $L$ and $S$ such that $|L|=O(m)$, $L$ contains all bins of
	weight at least $\frac{1}{m}$ and $L,S$-\AE is asymptotically optimal.
\end{lemma}

\begin{proof}
	By \Cref{lem:deterministic-opt}, there exists an optimal deterministic
	strategy $\opt$. Using the flow equation, \Cref{lem:freqbound} can be
	rewritten as:
	\[ \sum_{i=1}^n \mathcal{R}_i^\opt \leq 2m + n. \]
	Because $\opt$ will never recycle an empty bin, each $\mathcal{R}_i^\opt
	\geq 1$.  Therefore, there can be at most $m$ bins with average recycle
	rates at least 3. Let $L$ be this set of bins, together with any bins of
	weight at least $\frac{1}{m}$, and we will construct a strategy $S$ that
	aggressively empties the remaining bins into $L$.
	
	$S$ aggressively empties the complement of $L$, but also keeps a virtual
	configuration of where $\opt$ thinks the balls are, as well as a log of
	where $S$ has moved them. So when $S$ aggressively empties a bin, it also
	updates the log of each ball it throws, indicating where it landed. When
	$L^c$ is empty, it asks $\opt$ which bin to recycle based on the virtual
	configuration. If it says to recycle a bin in $L^c$, we use the logs to
	update where those balls will land in the virtual configuration. If it says
	to recycle a bin in $L$, we recycle those balls that are there in the
	virtual configuration, and leaving any others behind in that same bin.
	Thus $S$ performs $\opt$ but rushes ahead to recycle those balls outside of
	$L$.

	Now, consider $t$ rounds of $\opt$. For large enough $t$, $\opt$ will
	recycle on average at most 3 balls at a time from $L^c$. $S$ recycles at
	least 1 ball at a time from $L^c$ and exactly as many balls at a time from
	$L$.  Therefore for large $t$, $t$ rounds of $\opt$ will correspond to at
	most $3t$ rounds of $S$, and during this period $S$ will recycle the same
	number of balls. Thus $S$ is $1/3$-optimal.
\end{proof}

Next we compute the recycle rate of $L,S$-\AE as a function of the recycle rate
of $S$ on the induced ball-recycling game on $L$.

\begin{lemma}\label{lem:aggro-empty-bounds}
	If $\mathcal{R}^S$ is the recycle rate of $S$ (on $L$), and $q$ is the
	probability of a ball landing in $L^c$, then the recycle rate of $L,S$-\AE
	is
	\[\mathcal{R}^\textrm{AE} = \Theta\left(\frac{1}{(1 - q)/\mathcal{R}^S + q}\right)\]
\end{lemma}

\begin{proof}
	Consider a collection of recycling rounds of $L,S$-\AE where $t$ of those
	times $L,S$-\AE recycles a bin from $L$. Say $b$ balls are thrown from bins
	in $L$ and $a$ balls land in $L^c$. Now, if $m$ balls are thrown into bins
	of size at most $\frac{1}{m}$, then the expected number of empty bins is at
	most \[ m \left(1 - \frac{1}{m}\right)^m \leq \frac{m}{e}.\] Because fewer
	thrown balls will have fewer collisions, this means the expected number of
	non-empty bins when $k \leq m$ balls are thrown into $L^c$ is at least
	$\left(1-\frac{1}{e}\right)k$, requiring at least as many time steps to
	aggressively empty. Thus, for large $t$, the expected number of turns
	required to empty the $a$ balls out of $L^c$ is at least
	$\left(1-\frac{1}{e}\right)\frac{a}{1-q}$. Whereas even if the balls were
	recycled from $L^c$ one at a time this expected number of turns is at most
	$\frac{a}{1 - q}$ turns. The number of balls recycled during this period is
	$b + \frac{a}{1 - q}$, and we have shown the number of rounds $\rho$
	satisfies: \[\rho = \Theta\left(t + \frac{a}{1-q}\right).\]

	For large enough $t$, $b = \Theta\left(t\mathcal{R}^S\right)$ and
	$a = \Theta\left(tq\mathcal{R}^S\right)$, so the overall recycle rate
	$\mathcal{R}^\textrm{AE}$ therefore satisfies
	\begin{align*}
		\mathcal{R}^\textrm{AE} &= \Theta\left(\frac{t\mathcal{R}^S + tq\mathcal{R}^S/(1 - q)}{t + tq\mathcal{R}^S/(1 - q)}\right) \\
		&= \Theta\left(\frac{1}{(1 - q)/\mathcal{R}^S + q}\right)
	\end{align*}
\end{proof}

\subsecput{rboptproof}{\RB is Optimal}

In this section we will further examine the performance of \RB and show that it
is asymptotically optimal. We first describe a sufficient condition for
optimality of a strategy based on its recycle rate on $L$, then show that \RB
satisfies this criterion.

\begin{lemma}\label{lem:rl-opt}
	Let $L$ be a set of $O(m)$ bins for which there exists a strategy $T$ such
	that $L,T$-\AE is asymptotically optimal.  Let $\mathcal{R}^{\opt_L}$ be
	the recycle rate of the optimal strategy on the induced ball-recycling game
	of $L$.  For a given strategy $S$, let $\mathcal{R}_L^S$ be the conditional
	recycle rate of $S$ in the stationary distribution given that a ball in $L$
	is selected, and $q$ be the probability that a ball lands in $L^c$, i.e. $q
	= \sum_{k \in L^c} \p_k$. If either
	\[ \E{\mathcal{R}_L^S} = \Omega(\mathcal{R}^{\opt_L})
	\quad\textrm{or}\quad\E{\mathcal{R}_L^S} =
	\Omega\left(\frac{1}{q}\right),\]
	then $S$ is asymptotically optimal.
\end{lemma}

\begin{proof}
	By applying \Cref{lem:gen-flow-equation}, the subset variant of the
	flow equation, to $L$, 
	\[ f_L \mathcal{R}_L^S = (1-q)(f_L \mathcal{R}_L^S + (1-f_L) \mathcal{R}_{L^c}^S), \]
	where $f_L$ is the stationary probability of $S$ picking a bin in $L$.
	Solving for $f_L$,
	\begin{align} 
		f_L = \frac{\mathcal{R}_{L^c}^S}{q\mathcal{R}_L^S +
          \mathcal{R}_{L^c}^S}. \label{eqn:fL}
	\end{align}

	Suppose $\mathcal{R}_L^S = o\left(\frac{1}{q}\right)$ and $\mathcal{R}_L^S$
	is $\Theta(1)$-optimal on $L$. If $\mathcal{R}_L^S\leq \frac{1}{q}$, then
	$f_L \geq \frac{1}{2}$, and so because $\mathcal{R}^S = f_L\mathcal{R}_L^S
	+ (1-f_L)\mathcal{R}_{L^c}^S$, we must have $\mathcal{R}^S =
	\Omega(\mathcal{R}_L^S)$.

	Now, using \Cref{lem:aggroopt}, let $L$ and $T$ be such that $L,T$-\AE is
	asymptotically optimal, and let $\mathcal{R}^\textrm{AE}$ be its expected
	recycle rate. By \Cref{lem:aggro-empty-bounds},
	\begin{align*} 
		\mathcal{R}^\mathrm{AE} &= \Theta\left( \frac{1}{(1-q)/\mathcal{R}^T + q}\right) \\
		&= O\left(\mathcal{R}^T\right) = O\left(\mathcal{R}_L^S\right) = O\left(\mathcal{R}^S\right),\label{eqn:RAE}
	\end{align*} 
	so $S$ must be asymptotically optimal.

	If $\mathcal{R}_L^S = \Omega\left(\frac{1}{q}\right)$, then
	$\mathcal{R}_L^S > \frac{\alpha}{q}$ for some $\alpha$. Rearranging
	Equation~\eqref{eqn:fL} and multiplying by
	$\mathcal{R}_L^S$ yields
	\[ f_L \mathcal{R}_L^S = \frac{1}{\frac{q}{\mathcal{R}_{L^c}^S} + \frac{1}{\mathcal{R}_L^S}}. \]
	Here $\frac{1}{\mathcal{R}_L^S} \leq \frac{q}{\mathcal{R}_{L^c}^S}$, so
	$f_L\mathcal{R}_L^S = \Omega\left(\frac{1}{q}\right)$, and thus
	$\mathcal{R}^S = \Omega\left(\frac{1}{q}\right)$ as well. Now we can
	compare to $L,T$-\AE as above:
	\[ 	\mathcal{R}^\mathrm{AE}
		= \Theta\left( \frac{1}{(1-q)/\mathcal{R}^T + q}\right)
		= O\left(\frac{1}{q}\right)
		= O\left(\mathcal{R}^S\right)
		\]
	so in this case $S$ is asymptotically optimal as well.
\end{proof}	

We can now prove \Cref{thm:random-opt}.

\begin{proof}[Proof of \Cref{thm:random-opt}]
	If $m=\Omega(n)$, then by \Cref{lem:upperbound,lem:random-ball-lower} we are done.

	Otherwise, let $L$ be a set of $O(m)$ bins for which there exists a
	strategy $T$ such that $L,T$-\AE is asymptotically optimal. We will prove
	the result for a slightly modified \RB that only recycles 1 ball outside of
	$L$ even if more are available; that is, it moves only one of the balls in
	the bin. Since this strategy is worse than \RB, this will be sufficient. We
	number the bins so that the first $|L|$ bins comprise $L$.
	
	If $\mathcal{R}_L \geq \frac{1-q}{q}$, then we are done by
	\Cref{lem:rl-opt}. Otherwise, in the stationary distribution, when a
	bin in $L$ is recycled, the expected number of balls which land in $L^c$ is
	$q\mathcal{R}_L < 1-q$. When a bin in $L^c$ is recycled, the expected
	number of balls which land in $L$ is $1-q$. Thus \RB must pick a bin in $L$
	more than half the time, and so the expected number of balls in $L$ must be
	more than $\frac{m}{2}$.

	Now analogously to the proof of \Cref{lem:random-ball-lower}, we have:
	\begin{align*}
		\p_i\left(\sum_{j=1}^{|L|} \E{\left(\chi_j^\mathrm{RB}\right)^2} + \sum_{j=|L|+1}^n \E{\chi_j^\mathrm{RB}}\right) &= \E{\left(\chi_i^\mathrm{RB}\right)^2} \\
		&\geq \E{\chi_i^\mathrm{RB}}^2.
	\end{align*}

	Thus,
	\[ \E{\chi_i^\mathrm{RB}} \leq \sqrt{\p_i}\left(\sum_{j=1}^{|L|} \E{\left(\chi_j^\mathrm{RB}\right)^2} + \frac{m}{2}\right)^{\frac{1}{2}}.\]
	Summing over $i \leq |L|$ yields
	\[ \E{\chi_L^\mathrm{RB}} \leq \left(\sum_{i=1}^{|L|} \sqrt{\p_i}\right)\left(\sum_{j=1}^{|L|} \E{\left(\chi_j^\mathrm{RB}\right)^2} + \frac{m}{2}\right)^\frac{1}{2}, \]
	where $\chi_L^\mathrm{RB}$ is the expected number of balls in $L$. Now,
	\begin{align*}
		\mathcal{R}_L^\mathrm{RB} &\geq \frac{1}{m}\sum_{j=1}^{|L|} \E{\left(\chi_j^\mathrm{RB}\right)^2} \\
								  &\geq \frac{\E{\chi_L^\mathrm{RB}}^2}{m\left(\sum_{i=1}^{|L|} \sqrt{\p_i}\right)^2} - \frac{1}{2} \\
								  &> \frac{m}{4\halfnorm{\mathbf{p}_L}} - \frac{1}{2},
	\end{align*}
	where $\p_L$ is the conditional probability distribution on $L$ obtained
	from $\mathbf{p}$. The last inequality holds because there are at least
	$\frac{m}{2}$ balls in $L$ in expectation.

	Thus by \Cref{lem:upperbound}, \RB is asymptotically optimal on the induced
	system of $L$, and therefore \RB is asymptotically optimal by
	\Cref{lem:rl-opt}.
\end{proof}


%% file: uniform.tex

\secput{uniform}{The Uniform Case}

The results of \Cref{sec:nonuniform} hold for any distribution of the balls
into the bins. In this section we consider the special case where they are
uniformly distributed, which models insertion buffers as discussed in
\Cref{sec:motivation}.  We then show that \GG and \FB are optimal, up to
lower-order terms, in this setting, whereas \RB is at least $1/2$- and at most
$(1-\epsilon)$-optimal, for some constant $\epsilon>0$.

For a ball-recycling game with uniformly distributed balls,
\Cref{lem:upperbound} implies:
\begin{corollary}\label{cor:uniformupperbound}
	Consider a ball-recycling game with $m$ balls, $n$ bins and uniform
	distribution $\uni$. For any recycling strategy $A$, 
	\[\mathcal{R}^A \leq \frac{2m+n-1}{n} < 2\frac{m}{n} + 1.\]
\end{corollary}

The average number of balls in a bin is $m/n$, so \Cref{cor:uniformupperbound}
suggests that any ``reasonable'' strategy will be at least $1/2$-optimal in the
uniform case.

We now show that \GG and \FB are within an additive constant of optimal on strictly
uniform distributions.

\begin{lemma}\label{lem:unilowerbound}
	\GG and \FB each recycle at least $2m/(n+1)$ balls per round in
	expectation.
\end{lemma}

\begin{proof}
	Let $S$ be the random variable denoting the number of balls thrown in a
	given round with \GG. \GG will recycle the bins in order starting from the
	next one and cycling around. Therefore, we can consider the collection of
	bins to be a queue. After throwing the balls, the average place in the
	queue in which a ball lands is the $\left[(n-1)/s\right]$th bin, due to
	uniformity. Each ball thrown will therefore sit for an average of at most
	$(n-1)/2$ rounds before it is thrown again. Therefore,
	$m-\E{S}\leq\E{S}(n-1)/2$, and we have the result after solving for
	$\E{S}$.

	Let $T$ be the random variable denoting the number of balls thrown in a
	given round with \FB. If after removing the balls in the \FB, we list the
	bins in order of fullness, we can again think of the bins as a sort of
	queue.  When we throw the balls, the average place in the queue which a
	ball lands is the $\left[(n-1)/2\right]$th bin as above, due to uniformity.
	Now, we reorder the bins back into fullness order.  During the reordering
	more balls are moved up the queue than down, thus each ball thrown into the
	system will sit for an average of less than $(n-1)/2$ rounds before it is
	thrown again. Therefore, as above, $m-\E{S}\leq\E{S}(n-1)/2$, and we are
	done.
\end{proof}

\Cref{cor:uniformupperbound} and \Cref{lem:unilowerbound}
together prove \Cref{thm:fullestbin}. Despite these strong
performance bounds, recall that \FB can perform arbitrarily badly on
non-uniform $\p$. \RB on the other hand is always $\Theta(1)$-optimal.

\subsection{\RB in the Uniform Case}

However, \RB{} does not achieve this level of optimality on uniform
distributions. In this section we will show in \Cref{thm:rbuniform} that \RB
recycles at most $1+(2-\epsilon)m/n$ balls per round in expectation, for some
$\epsilon > 0$.  The upper bound is given in \Cref{lem:rbuni-upper} and
\Cref{cor:rbuni-upper-constants}, and the lower bound is given in
\Cref{lem:rbuni-lower}.

We begin with the following lemma:

\begin{lemma}\label{lem:pair-flow}
	Let $\chi^\textrm{RB}$ be the stationary distribution relative to \RB,
	$R^\textrm{RB}(X)$ the random variable of how many balls \RB recycles from
	ball configuration $X$, and $\mathcal{R}^\textrm{RB} =
	\E{R^\textrm{RB}\left(\chi^\textrm{RB}\right)}$ the expected recycle rate
	of \RB. Then,
	\[ \frac{\E{R^\textrm{RB}\left(\chi^\textrm{RB}\right)^2}}{\mathcal{R}^\textrm{RB}} = \frac{2m+n-1}{n+1} \leq 1 + \frac{2m}{n}. \]
\end{lemma}

\begin{proof}
	Consider the random variable of the number of distinct unordered pairs of
	balls which are in the same bin in $\chi^\textrm{RB}$. In expectation, a
	round of \RB eliminates 
	\[\binom{\mathcal{R}^\textrm{RB}}{2}\]
	and creates
	\[\sum_{k=0}^{\mathcal{R} - 1} \frac{m - \mathcal{R}^\textrm{RB} + k}{n}\]
	such pairs. In the stationary distribution, these must be equal, so
	\begin{multline*}
		\frac{\E{R^\textrm{RB}\left(\chi^\textrm{RB}\right)^2}}{2} - \frac{\mathcal{R}^\textrm{RB}}{2} \\
		= \frac{(2m-1)\mathcal{R}^\textrm{RB}}{2n} - \frac{\E{R^\textrm{RB}\left(\chi^\textrm{RB}\right)^2}}{2n}.
	\end{multline*}
	After rearranging we have the result.
\end{proof}

\begin{lemma}\label{lem:rbuni-upper}
	There exists a constant $\alpha > 0$ such that \RB is at most
	$(1-\alpha)$-optimal.
\end{lemma}

\begin{proof}
	Let $\chi^\textrm{RB}$ be the stationary distribution relative to \RB,
	$R^\textrm{RB}(X)$ the random variable of how many balls \RB recycles from
	ball configuration $X$, and $\mathcal{R}^\textrm{RB} =
	\E{R^\textrm{RB}\left(\chi^\textrm{RB}\right)}$ the expected recycle rate
	of \RB.  We will prove the result by contradiction, so assume that for all 
	constant $\epsilon > 0$
	\[\mathcal{R}^\textrm{RB} \geq 1 + \frac{(2-\epsilon)m}{n}.\]

	Let $c\in(1,2)$ be a constant to be determined later. We say a bin is
	\defn{light} if it contains at most $cm/n$ balls. Let $L$ be the random
	variable of the number of balls in light bins in the stationary
	distribution. Then the probability $q_L$ that \RB recycles a light bin in
	the stationary distribution is $\E{L}/m$. We proceed by cases.

{\bf Case 1.} Suppose $\E{L} \geq \delta m$ for some constant $\delta > 0$. Then $q_L \geq \delta$ and for $c \leq 2 - 2\epsilon$ and $\epsilon < 1/2$,
\begin{align*}
	\Var{R^\textrm{RB}\left(\chi^\textrm{RB}\right)} 
	&= \E{\left(R^\textrm{RB}\left(\chi^\textrm{RB}\right) - \mathcal{R}^\textrm{RB}\right)^2}  \\
	&\ge q_L \left(1+\frac{(2-\epsilon)m}{n}-\frac{cm}{n}\right)^2 \\
	&\ge \frac{\epsilon^2\delta}{4} \left(\frac{4m^2}{n^2}+\frac{4m}{n}+1\right) \\
	&\ge \frac{\epsilon^2 \delta}{4} \left(\mathcal{R}^\textrm{RB}\right)^2. \label{eqn:varFl}
\end{align*}

Thus by the definition of variance, we have
\begin{equation*}
	\E{R^\textrm{RB}\left(\chi^\textrm{RB}\right)^2} \ge \left(1+\frac{\epsilon^2\delta}{4}\right)\left(\mathcal{R}^\textrm{RB}\right)^2. \label{eqn:expF2l}
\end{equation*}
	Now by \Cref{lem:pair-flow},
\begin{equation*}
	\mathcal{R}^\textrm{RB} \le \left(1+\frac{\epsilon^2\delta}{4}\right)^{-1}\left(1+\frac{2m}{n}\right).
\end{equation*}
Since $\epsilon, \delta$ are constants greater than $0$, we have our contradiction for
	the first case. 

\medskip 

{\bf Case 2.} Otherwise, $\E{L} < \delta m$. Since $L \in [0, m]$, $\E{L^2} <
	\delta m^2$. \Cref{lem:pair-flow} implies $\E{R^\textrm{RB}\left(\chi^\textrm{RB}\right)^2} \le (1+2m/n)^2$. Together H\"{o}lder's inequality we have
\begin{align}
	\E{LR^\textrm{RB}\left(\chi^\textrm{RB}\right)}
	&\le \left(\E{L^2}\E{R^\textrm{RB}\left(\chi^\textrm{RB}\right)^2}\right)^{1/2} \nonumber \\
	&< \left(\delta m^2 \left(1+\frac{2m}{n}\right)^2\right)^{1/2} \nonumber\\
	&= \sqrt{\delta} m\left(1 + \frac{2m}{n}\right) \label{eqn:expLFu}
\end{align}

	Let $Y$ be the random variable of the number of balls in the stationary
	distribution which start in a light bins, but end up begin among the first
	$1 + cm/n$ balls in a heavy bin after an application of \RB.  Let $\Phi$
	be the random variable of the number of distinct unordered pairs of balls
	that are in the same light bin in the stationary distribution. Applying \RB
	in expectation creates at most
	\begin{multline*}
          \E{\sum_{k=0}^{\mathcal{R}^\textrm{RB}-1}\frac{L + k}{n}} \\
		  = \E{\frac{2LR^\textrm{RB}\left(\chi^\textrm{RB}\right) + R^\textrm{RB}\left(\chi^\textrm{RB}\right)^2 - R^\textrm{RB}\left(\chi^\textrm{RB}\right)}{2n}}
	\end{multline*}
	such pairs, and eliminates at least
	\[ \E{\frac{Y}{1+\frac{cm}{n}}\binom{1+\frac{cm}{n}}{2}}.\]
	In the stationary distribution these quantities must be equal, so
	rearranging together with Equation~\eqref{eqn:expLFu}, we have
	
\begin{align*}
	\E{Y} 
	&\le \frac{2\E{LR^\textrm{RB}\left(\chi^\textrm{RB}\right)} + \E{R^\textrm{RB}\left(\chi^\textrm{RB}\right)^2} - \mathcal{R}^\textrm{RB}}{cm} \nonumber \\
	&< \frac{2\sqrt{\delta}}{c}\left(1+\frac{2m}{n}\right) + \frac{\E{R^\textrm{RB}\left(\chi^\textrm{RB}\right)^2} - \mathcal{R}^\textrm{RB}}{cm} \\
	&<\left(1+\frac{2m}{n}\right)\left(\frac{2\sqrt{\delta}}{c}+\frac{2}{cn}\right),  
\end{align*}
where we have used \Cref{lem:pair-flow} for the last inequality.

We now compute the effect on $\E{L}$ of applying \RB to the stationary
distribution.  By Markov's inequality, there must be more than $(1-1/c)n$ light
bins, and so the probability that a ball is thrown into a light bin is more
than $1-1/c$. Therefore, at least $(1-1/c)\mathcal{R}^\textrm{RB}$ balls land
in light bins in expectation. We expect at most $\E{Y}$ balls to be in light
bins which turn into heavy bins. Finally, we recycle at most $cm/n$ balls from
a light bin $\E{L}/m$ of the time. Since the net change to $L$ must be $0$ in
expectation,
\[\left(1-\frac{1}{c}\right)\mathcal{R}^\textrm{RB} < \frac{c}{n}\E{L} + \E{Y}.\]

However, this is a contradiction. Indeed, the LHS is at least 
$$
\left(1-\frac{1}{c}\right)\left(1+\frac{(2-\epsilon)m}{n}\right),
$$
but the RHS is less than
$$
\delta c \frac{m}{n} + \left(1+\frac{2m}{n}\right)\left(\frac{2\sqrt{\delta}}{c}+\frac{2}{cn}\right).
$$
Thus, if we pick a sufficiently small $\delta > 0$, $\epsilon = 0.01$, $c=1.98$
and $n\geq 3$, we have a contradiction. For $n \le 2$, the contradiction follows
immediately from \Cref{lem:pair-flow}.
\end{proof}

\begin{corollary}\label{cor:rbuni-upper-constants}
	Setting
	\begin{equation*}
		(\epsilon, c, \delta) = (0.001, 1.456, 0.042)
	\end{equation*}
	in the proof of \Cref{thm:rbuniform}, we obtain 
	\begin{equation*}
		\mathcal{R}^\textrm{RB} < 1 + 1.994\frac{m}{n}.
	\end{equation*}
\end{corollary}

\begin{lemma}\label{lem:rbuni-lower} For all $c > 0$, there exists a
  $c'$ such that if $m \ge c' n\log n$, the uniform random ball policy
  has expected recycle rate at least   
$$
 	\left(1+\frac{1}{6^4} - c\right)\frac{m}{n}.
$$
\end{lemma}

\begin{proof}
	Let $X_{t, k}$ be the random variable denoting the number of balls in the
	$k$th bin at the beginning of the $t$th round. Because of symmetry,
	$X_{t, k}$ follows the same distribution as $X_{t, \ell}$ for any $k \ne
	\ell$.  For simplicity, we let $X_t$ be a random variable that follows the
	same distribution as $X_{t, k}$ for all $k$. 

	We pick $t$ to be sufficiently large so that the system enters its
	stationary state after $t$ rounds. Thus, $X_t$ and $X_{t'}$ follows the
	same distribution for any $t' > t$.

	Let $Y_t$ be the random variable denoting the number of balls recycled in
	the $t$th round. By definition, we have
	$$ \E{Y_t} = \sum_{1 \le k \le n}\frac{\E{X_{t, k}^2}}{m} =
	\frac{n}{m} \left( \E{X_t}^2 + \Var{X_t} \right).  $$
	Note that $\E{X_t} = m/n$, and so $\E{Y_t} \ge m/n$. 

	To show $\E{Y_t}$ deviates from $m/n$, we derive a lower bound for
	$\Var{X_t}$.

	If $\Prob{X_t \le (1-\epsilon)m/n} \ge \delta$, then $\E{Y_t} \ge
	(1+\epsilon^2\delta)m/n$.

	Otherwise $\Prob{X_t > (1-\epsilon)m/n } > 1-\delta$. We will show that if
	$\delta$ is small enough, then this case does not exist. 
		
	We say a bin is {\em heavy} if it has more than $(1-\epsilon)m/n$ balls.
	Let $Z_t$ be the random variable denoting the number of heavy bins at the
	beginning of the $t$th round. We have
	\begin{equation*}
		\E{Z_t} = \sum_{1 \le k \le n} \E{\ind{X_{t, k} >
		(1-\epsilon)\frac{m}{n}}} > (1-\delta) n.
	\end{equation*}
	$Z_t$ is a non-negative variable in $[0, n]$ and has expected value more
	than $(1-\delta)n$. By Markov's inequality,
	\begin{equation*}
		\Pr[Z_t \le (1-2\delta)n] < \frac{1}{2} \mbox{ and } \Pr[Z_t >
		(1-2\delta)n] > \frac{1}{2}.
	\end{equation*}

	We compute $\E{Z_{t+n/2}}$ from the $Z_t$. If $Z_t > (1-2\delta)n$
	for some constant $\delta < 1/4$, the following hold during $P$, the time
	period between the $t$th round and the $(t+n/2)$th round:

	\begin{enumerate}
		\item At least $(1/2-2\delta)n$ bins in $H_t$ are recycled, 
		\item At least $(1/2-2\delta)(1-\epsilon)m$ balls are recycled,
		\item At least $(1/2-2\delta)n$ bins in $H_t$ are not recycled,
	\end{enumerate}
	where $H_t$ denotes the set of heavy bins at the beginning of the $t$th
	round.

	Given (c), we can find a subset $S_t \subset H_t$ that is composed of
	$(1/2-2\delta)n$ bins in $H_t$ not recycled during $P$.  Note that which
	bins are recycled and which are not depend on the random choices made by the
	system. Hence, $S_t$ varies.  
		
	Next, we derive a lower bound on the expected number of balls in any $S_t$.
	The balls which stay in $S_t$ come from two different sources. There are
	those that stay in $S_t$ at the beginning of the $t$th round, of which
	there are at least $|S_t|(1-\epsilon)m/n$.  There are also those which are
	recycled during $P$, of which there are at least
	$|S_t|(1-\epsilon')|B|/n$ by \Cref{lem:MinBin}, to follow.  Combining the
	two sources, the expected number of balls in $S_t$ is at least
	\begin{equation*}
		\Gamma = (1-\epsilon)\left( \left(\frac{1}{2}-2\delta\right) +
		\left(\frac{1}{2}-2\delta\right)^2(1-\epsilon')\right)m.
	\end{equation*}

	\begin{lemma}\label{lem:MinBin}
	Let $B$ be the multiset of the first $(1/2-2\delta)(1-\epsilon)m$ balls
	recycled during $P$. $B$ is well-defined thanks to 2.\ above. Let $L_i$ be
	the random variable denoting the number of balls in $B$ that land on the
	$i$th bin. For all $\epsilon' > 0$, there exists a $c'$ such
        that if $m \ge c' n \log n$
		\begin{equation*}
			\E{\min\{L_1, L_2, \ldots, L_n\}} \ge (1-\epsilon')|B|/n.
		\end{equation*}
	\end{lemma}
	\begin{proof}
		For $i \in [1,n]$, $\E{L_i} = |B|/n$.
		By Chernoff bounds, 
		$$
			\Pr[|L_i-\E{L_i}| \ge (\epsilon'/2) \E{L_i}] \le \frac{1}{n^2} 
		$$
		for some sufficiently large $c'$. Consequently, by the union bound,
		$$
			\Pr[\min\{L_1, L_2, \ldots, L_n\} \le (1-\epsilon'/2)|B|/n] \le \frac{1}{n}.
		$$
		Because the $L_i$'s are non-negative, we are done. 
	\end{proof}

	Given $\Gamma$, we obtain the following bound:
	\begin{multline*}
		\E{Z_{t+n/2} \given Z_t > (1-2\delta)n} \le |S_t| + \frac{m-\Gamma}{(1-\epsilon)m/n} \\
		\begin{aligned} &= \left(\frac{1}{1-\epsilon} - \left(\frac{1}{2}-\delta\right)^2(1-\epsilon')\right) n \\
		&\approx \left(\frac{3}{4}+\epsilon+\delta\right) n \end{aligned}
	\end{multline*}
	Together with the trivial bound $\E{Z_{t+n/2} \given Z_t \le (1-2\delta)n}
	\le n$,  $\E{Z_{t+n/2}}$ equals
	\begin{multline*}
		\Prob{Z_t \le (1-2\delta)}\E{Z_{t+n/2} \given Z_t \le (1-2\delta)} \\
		+ \Prob{Z_t > (1-2\delta)}\E{Z_{t+n/2} \given Z_t > (1-2\delta)} \\
		\begin{aligned} &< \frac{1}{2} \left( 1 + \left(\frac{3}{4}+\epsilon+\delta\right) \right) n\\
		&\approx \left(\frac{7}{8} + \frac{\epsilon+\delta}{2}\right) n \end{aligned}
	\end{multline*}
		
	This leads to a contradiction if $\epsilon+\delta$ is small enough.  This
	is because we have $\E{Z_t} > (1-\delta)n$ and $\E{Z_{t+n/2}} = \E{Z_t}$,
	because the system is stationary.  As a result, we have a contradiction if
	$\frac{\epsilon+3\delta}{2} < \frac{1}{8}$.

	Combining the results for the two cases, we wish to maximize
	$1+\epsilon^2\delta$ subject to $\epsilon+3\delta < \frac{1}{4}.$ Picking
	$\epsilon = 1/6$ yields the result.
\end{proof}

\Cref{thm:rbuniform} follows from
\Cref{cor:rbuni-upper-constants,lem:rbuni-lower,cor:uniformupperbound}.


%% file: experiments.tex

\section{Database Experiments}\label{sec:experiments}
In this section, we consider insertion buffers as they are used in practice. We
demonstrate through simulations as well as experiments on real-world systems,
that the theoretical results in the prior sections hold and can be used to
improve performance.

\subsection{Insertion Buffers in Database Systems}
Many databases cache recently inserted items in RAM so that they can write
items to disk in batches.  Examples include
Azure~\cite{Azure16}, 
DB2~\cite{IBM17},  
Hbase~\cite{Xiang12}, 
Informix~\cite{Informix}, 
InnoDB~\cite{Callaghan11},  
NuDB~\cite{NuDB16},  
Oracle~\cite{Oracle17},  
SAP~\cite{SAP17},  and 
Vertica~\cite{Vertica17}.
They are also used to accelerate inserts in several research prototypes, such
as the buffered Bloom
filter~\cite{CanimLaMi10} and buffered quotient
filter~\cite{BenderFaJo12}.  
By batching updates to disk, these insertion buffers reduce the amortized
number of I/Os per insert, which can substantially improve insertion
throughput.  Facebook claims that the insertion buffer in InnoDB speeds up some
production workloads by a factor of $5$ to $16$, and accelerates some synthetic
benchmarks by up to a factor of $80$~\cite{Callaghan11}.

A motivating factor for the use of insertion buffers is that they can
significantly mitigate the precipitous performance drop that databases can
experience when the data set grows too large to fit in RAM.
\Cref{fig:introbufferinitial} shows the time per 1,000 insertions into a MySQL
database using the InnoDB backend, with and without InnoDB's insertion buffer
enabled.  For the first $200,000$ insertions, the entire database fits in RAM,
and so insertions are fast, even without the insertion buffer.

Once the database grows larger than RAM, insertion performance without the
insertion buffer falls off a cliff.  In fact, once the database reaches 1M
rows, it can perform only about 200 insertions per second, suggesting that the
throughput is limited by the random-I/O performance of the underlying disk.  In
the benchmark with the insertion buffer enabled, on the other hand, performance
degrades by only a small amount.

Based on the performance of the first 1M insertions, it appears that InnoDB's
insertion buffer effectively eliminates the performance cliff that can occur
when the database grows larger than RAM.  This improvement explains the
popularity of insertion buffers in database design.

However, in our experiment, as the database continues to grow, the efficacy of
the insertion buffer declines.  \Cref{fig:introbuffer} shows the time per
10,000 insertions as the database grows to 50M rows.  Although the performance
without the insertion buffer drops more quickly early on, it remains relatively
stable thereafter.  Performance with the insertion buffer, on the other hand,
slowly declines over the course of the benchmark until it is only about a third
faster than without the insertion buffer.  This is well below the $5-80\times$
speedups reported above.

\input{intro-plots}

As these experiments show, it can be difficult to extrapolate from small
examples the performance gains that insertion buffers can provide for large
databases. Therefore, it is no wonder that reported speedups from insertion
buffers vary wildly from as little as $2\times$ to as high as
$80\times$~\cite{Callaghan11}.  Some have even suggested that insertion buffers
may provide many of the benefits of write-optimization~\cite{Callaghan10},
i.e., that insertion buffers can bring the performance of \btrees{} up to that
of LSM-trees~\cite{ONeilChGa96}, COLAs~\cite{BenderFaFi07}, Fractal
Trees~\cite{Tokutek14}, xDicts~\cite{BrodalDeFi10}, or
\betrees{}~\cite{BrodalFa03}.

\subsection{Experimental Validation}
Here we validate our theoretical study of insertion buffers by showing that our
analysis above can have a material impact on the performance of databases with
insertions buffers.  We simulated workloads of random insertions to a \btree{},
with varying distributions on the inserted keys.  We found that, as predicted,
the performance was independent of the input distribution and closely matched
the performance predicted by our theorems.

We then ran workloads of random insertions into InnoDB and measured the average
batch size of flushes from its insertion buffer.  InnoDB implements a variant
of the random-item flushing strategy.  We modified it to implement the
golden-gate flushing strategy.  Despite the additional complexities of InnoDB's
insertion buffer implementation, we found that performance closely tracked our
theoretical predictions and was independent of the distribution of inserted
keys.  We also found that the golden-gate flushing strategy improved InnoDB's
flushing rate by about 30\% over the course of our benchmark.

Our analysis explains why insertion buffers can provide dramatic speedups for
small databases, but only small gains are available as the database grows.  Our
results also provide useful guidance to implementers about which flushing
strategy will provide the most performance improvement.

Our results also show that insertion buffers cannot deliver the same asymptotic
performance improvements that are possible with write-optimized data
structures, such as LSM-trees and B$^\epsilon$-trees.

\subsecput{exp-deployed}{Insertion-Buffer Background}

This section describes insertion buffers are actually implemented and used in
deployed systems and recent research prototypes.

\paragraph{SAP:} The SAP IQ database supports an in-memory row-level versioning
(RLV) store, and insertions are performed to the RLV store and later merged
into the main on-disk store~\cite{SAP17}.  

\paragraph{NuDB:}  The NuDB SSD-based key-value store buffers all insertions in
memory, and later flushes it to SSD~\cite{NuDB16}.  Flushes occur at least once
per second, or more often if insertion activity causes the in-memory buffer to
fill.

\paragraph{Buffered Bloom and quotient filters:}  Bloom filters are known to
have poor locality for both inserts and lookups.  The buffered Bloom
filter~\cite{CanimLaMi10} improves the performance of insertions to a Bloom
filter on SSD by buffering the updates in RAM.  The on-disk Bloom filter is
divided into pages, and each page has a buffer of updates in RAM.  When a
page's buffer fills, the buffered changes are written to the page.  

The buffered quotient filter stores newly inserted items in an in-memory
quotient filter~\cite{BenderFaJo12,benderadaptivebloom}.  When the in-memory
quotient filter fills, its entire contents are flushed to the on-disk quotient
filter.

\paragraph{InnoDB:}  The InnoDB~\cite{Oracle17a} \btree{} implementation used
in the MySQL~\cite{Oracle17b} and MariaDB~\cite{Foundation17} relational
database systems includes an insertion buffer.

Our experiments in this paper focus on InnoDB as an archetypal and open-source
implementation of an insertion buffer, so we describe it in detail.

InnoDB structures its insertion buffer as a \btree{}.  When the insertion
buffer becomes full, it selects the items to be flushed by performing a random
walk from the root to a leaf.  The random walk is performed by selecting, at
each step, uniformly randomly from among the children of the current node.
Once it gets to a leaf, it it picks a single item to insert into the on-disk
\btree{}.  This item, along with any other items in the insertion buffer that
belong in that leaf, are inserted into the leaf and removed from the insertion
buffer.

InnoDB's insertion buffer is complicated in several ways.  First, the size of
the insertion buffer changes over time, as InnoDB allocates more or less space
to other buffers and caches.

InnoDB also has a leaf cache.  Whenever a leaf is brought into cache for any
reason, all inserts to that leaf that are currently in the insertion buffer are
immediately applied to the leaf, and any future inserts to that leaf also skip
the insertion buffer as long as the leaf remains in cache.

Finally, it performs some flushing when the buffer is not full. Roughly every
second, InnoDB performs a small amount of background flushing. Moreover, it
prematurely flushes its buffer to a leaf when it calculates that such a flush
will cause the leaf to split. We hypothesize that this feature exists to
simplify the transactional system.

\subsecput{exp-insertion}{Leaf Probabilities in \btrees{}}
In \Cref{sec:uniform}, we established that, on insertion, the leaf
probabilities are nearly uniform.  We empirically verify this uniformity
property by simulating insertions into the leaves of a \btree{}. We insert
real-valued keys i.i.d.\ according to uniform, Pareto (real-valued Zipfian) and
normal distributions; the leaves of the \btree{} split when they are full, and
we measure the ratio of the maximal weight leaf to $1/n$.
\Cref{lem:uniform-leaves} tells us that this ratio should be asymptotically at
most constant, but as \Cref{fig:simalpha} shows, our experimental analysis
shows further that this constant is generally less than 2. Because leaves
generally split in 2, this makes some intuitive sense.

\input{insertion-plots}

We also verify the these results using the \innodb{} storage engine. We insert
5 million rows into a database using uniform, Pareto and normal distributions
on the keys. the results are summarized in \cref{fig:alpha}. The maximum ratio
does not exceed 2.3, and the 95th percentile ratio does exceed 1.6. Thus the
distribution of the keys to the leaves is in fact almost uniform.

\subsecput{exp-synthetic}{Simulating Insertion Buffers}

The ball-and-bins models described above are based on a static leaf structure.
However, in practice inserting into a database causes the leaf structure (the
number and probability distribution of bins in the model) to change. However,
we can still perform the same strategies, and by simulating an insertion buffer
in front of a database, we can compare their efficiency as well as verify that
much of the static analysis empirically applies to the dynamic system.

We insert real-valued keys into the simulation according to one of several
distributions of varying skewness: uniform on $[0,1000]$, Pareto with parameter
$\alpha=\{0.5,1.0,2.0\}$, and uniform centered at 0, with standard deviation
$1000$. We have a buffer which stores 2,500 keys; when it fills we choose a
leaf according to the chosen strategy and flush all the buffered keys destined
to it. Initially we have one leaf, and the leaves split when they exceed 160
keys, as uniformly as possible. 

As shown in \Cref{fig:synth}, the key distribution doesn't affect the recycle
rate of the insertion buffer, and as the number of leaves gets larger, the
recycle rate decreases. Generally fullest bin does better than golden gate, and
golden gate does better than random ball. Demonstrated with the normal
distribution (all distributions perform very similarly), \Cref{fig:synth-ratio}
shows that golden gate initially outperforms random ball by about 30\%, which
then decreases as the number of bins grows.

\input{synthetic-plots}

\subsecput{exp-innodb}{Real-World Performance (InnoDB)}

In this section, we empirically the performance of insertion buffers in
\innodb{}, the default storage engine in \mysql{}.

Analogously to the experiments in \secref{exp-synthetic}, we insert rows into
the \mysql{} database, and after every 10000 insertions, we check the ``merge
ratio'' reported by \innodb{}.  This is the number of rows merged into the
database from the buffer during each buffer flush, and corresponds to the
recycling rate in the balls and bins model. We also check the reported memory
allocated to the buffer, which allows us to control for memory usage. 

\input{innodb-plots}

The keys of the rows are i.i.d.\ according to the same real-valued probability
distributions as in \secref{exp-synthetic}: uniform on $[0,1000000]$, Pareto
with parameter $\alpha = \{0.5,1,2\}$, and normal centered at 0 with standard
deviation $1000$. The results for the different distributions are shown in
\Cref{fig:innodbuni,fig:innodbpar0.5,fig:innodbpar1,fig:innodbpar2,fig:innodbnorm}.
The structure of the plot generally does not depend on the key distribution,
and while there is more noise, the overall picture is similar to the plots in
\Cref{fig:synth}. 

If we were to hold the number of leaves roughly constant and change the buffer
size, \Cref{lem:unilowerbound} suggests that the relationship with recycle rate
would be roughly linear. To test this, we ran the above experiment with buffer
sizes from 8mb to 128mb in 2mb increments. We performed 11 million insertions
with uniformly distributed keys each time, and then took the average recycle
rate for the last million rows. As demonstrated in \Cref{fig:innodbbuffer}, the
resulting plot is approximately linear.


%% file: intro-plots.tex

\begin{figure*}
{\centering
\ref{intro_legend}
\subfloat[First million insertions]{
\tikzsetnextfilename{intro-buffer-initial}
\begin{tikzpicture}
	\begin{axis}[
		width=0.5\textwidth,
		xlabel=Rows inserted (in millions),
		ylabel=Seconds per thousand insertions,
		ylabel near ticks,
		xmin=0,
		xmax=1,
		ymin=0,
		mark repeat={2},
		legend to name=intro_legend,
		legend columns=2,
		]
		\addplot [mark=triangle*, YellowOrange, mark size=2] table [x expr=\thisrowno{0}/1000000, y expr=\thisrowno{1} / 10000, col sep=comma] {./data/intro/buffer_off/output.txt};
		\addlegendentry{Insertion buffer disabled}
		\addplot [mark=pentagon*, RoyalPurple, mark size=2] table [x expr=\thisrowno{0}/1000000, y expr=\thisrowno{1} / 10000, col sep=comma] {./data/intro/buffer_on/output.txt};
		\addlegendentry{Insertion buffer enabled}
	\end{axis}
\end{tikzpicture}
\label{fig:introbufferinitial}
}
\subfloat[All fifty million insertions] {
\tikzsetnextfilename{intro-buffer}
	\begin{tikzpicture}
	\begin{axis}[
		width=0.5\textwidth,
		xlabel=Rows inserted (in millions),
		ylabel=Seconds per thousand insertions,
		ylabel near ticks,
		xmin=0,
		xmax=50,
		ymin=0,
		mark repeat={5},
		]
		\addplot [mark=triangle*, YellowOrange, mark size=2, each nth point={10}] table [x expr=\thisrowno{0}/1000000, y expr=\thisrowno{1} / 10000, col sep=comma] {./data/intro/buffer_off/output.txt};
		\addplot [mark=pentagon*, RoyalPurple, mark size=2, each nth point={10}] table [x expr=\thisrowno{0}/1000000, y expr=\thisrowno{1} / 10000, col sep=comma] {./data/intro/buffer_on/output.txt};
	\end{axis}
	\end{tikzpicture}
\label{fig:introbuffer}
}
	\caption{The cost of inserting batches of rows into an empty
          table in \innodb{} with and without the insertion
          buffer. The rows are inserted in batches of 10,000 to avoid
          slowdown in parsing, and the keys are distributed
          uniformly. After 1M insertions, the buffered version takes
          $12.3\%$ as long as the unbuffered version (measured over
          50000 insertions); after 50M insertions, the advantage is
          reduced so that the buffered version takes $68.3\%$ of
          the time of the unbuffered. (Lower is better)}
}
\end{figure*}
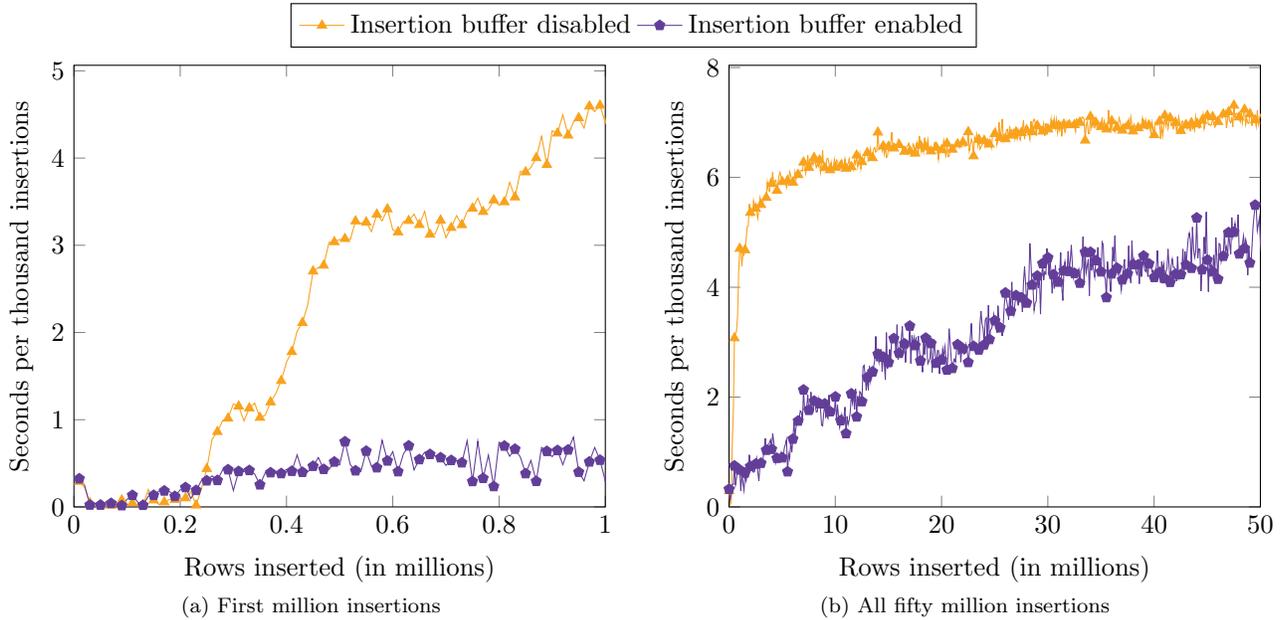


%% file: insertion-plots.tex

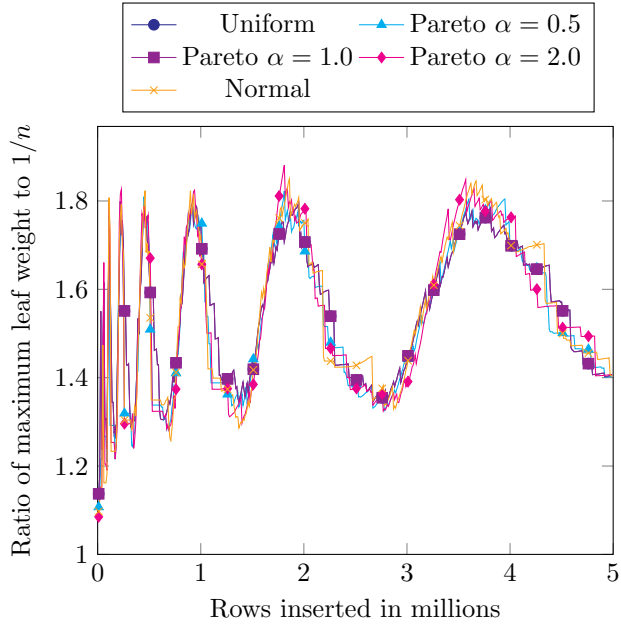
\begin{figure}[h!]
{\centering
\hspace*{33pt}\ref{full_distro_legend}\\
\tikzsetnextfilename{insertion-sim-max}
	\begin{tikzpicture}
		\begin{axis}[
			xlabel=Rows inserted in millions,
			xlabel near ticks,
			ylabel=Ratio of maximum leaf weight to $1/n$,
			ylabel near ticks,
			xmin=0,
			xmax=5,
			ymin=1,
			mark repeat={25},
			legend to name=full_distro_legend,
			legend columns=2,
			]
			\addplot [mark=*, Blue, mark size=2] table [x=insertions, y=max_alpha, col sep=comma] {./data/simulations/noop_uni.csv};
			\addlegendentry{Uniform}
			\addplot [mark=triangle*, Cyan, mark size=2] table [x=insertions, y=max_alpha, col sep=comma] {./data/simulations/noop_par5.csv};
			\addlegendentry{Pareto $\alpha=0.5$}
			\addplot [mark=square*, Plum, mark size=2] table [x=insertions, y=max_alpha, col sep=comma] {./data/simulations/noop_par1.csv};
			\addlegendentry{Pareto $\alpha=1.0$}
			\addplot [mark=diamond*, Magenta, mark size=2] table [x=insertions, y=max_alpha, col sep=comma] {./data/simulations/noop_par2.csv};
			\addlegendentry{Pareto $\alpha=2.0$}
			\addplot [mark=x, YellowOrange, mark size=2] table [x=insertions, y=max_alpha, col sep=comma] {./data/simulations/noop_norm.csv};
			\addlegendentry{Normal}
		\end{axis}
	\end{tikzpicture}
\caption{Deviation of the maximum weight leaf from uniform in simulation}
\label{fig:simalpha}
}
\end{figure}

\begin{figure}
{\centering
\hspace*{35pt}\ref{full_distro_legend}\\
\subfloat[Maximum leaf weight relative to uniform]{
\tikzsetnextfilename{insertion-inno-max}
\begin{tikzpicture}
\begin{axis}[
    width=0.47\textwidth,
	xlabel=Rows inserted in millions,
	xlabel near ticks,
	ylabel near ticks,
	ylabel=Ratio of maximum leaf weight to $1/n$,
	xmin=0,
	xmax=5,
	ymin=1,
	mark repeat={25},
	legend columns=2,
	legend to name=distro_legend
	]
	\addplot [mark=*, Blue, mark size=2] table [x=insertions, y=alpha, col sep=comma] {./data/innodb/randomball_uniform_1.0.csv};
	\addlegendentry{Uniform}
	\addplot [mark=triangle*, Cyan, mark size=2] table [x=insertions, y=alpha, col sep=comma] {./data/innodb/randomball_pareto_0.5.csv};
	\addlegendentry{Pareto $\alpha=0.5$}
	\addplot [mark=square*, Plum, mark size=2] table [x=insertions, y=alpha, col sep=comma] {./data/innodb/randomball_pareto_1.0.csv};
	\addlegendentry{Pareto $\alpha=1.0$}
	\addplot [mark=diamond*, Magenta, mark size=2] table [x=insertions, y=alpha, col sep=comma] {./data/innodb/randomball_normal_1000.0.csv};
	\addlegendentry{Pareto $\alpha=2.0$}
	\addplot [mark=x, YellowOrange, mark size=2] table [x=insertions, y=alpha, col sep=comma] {./data/innodb/randomball_normal_1000.0.csv};
	\addlegendentry{Normal}
\end{axis}
\end{tikzpicture}
\label{fig:alpha}
}\hfill
\subfloat[95th percentile leaf weight relative to uniform]{
\tikzsetnextfilename{insertion-inno-95}
\begin{tikzpicture}
\begin{axis}[
	width=0.47\textwidth,
	xlabel=Rows inserted in millions,
	xlabel near ticks,
	ylabel near ticks,
	ylabel=Ratio of 95th percentile leaf weight to $1/n$,
	xmin=0,
	xmax=5,
	ymin=1,
	mark repeat={25}
	]
	\addplot [mark=*, Blue, mark size=2] table [x=insertions, y=perc, col sep=comma] {./data/innodb/randomball_uniform_1.0.csv};
	\addplot [mark=triangle*, Cyan, mark size=2] table [x=insertions, y=perc, col sep=comma] {./data/innodb/randomball_pareto_0.5.csv};
	\addplot [mark=square*, Plum, mark size=2] table [x=insertions, y=perc, col sep=comma] {./data/innodb/randomball_pareto_1.0.csv};
	\addplot [mark=diamond*, Magenta, mark size=2] table [x=insertions, y=perc, col sep=comma] {./data/innodb/randomball_pareto_2.0.csv};
	\addplot [mark=x, YellowOrange, mark size=2] table [x=insertions, y=perc, col sep=comma] {./data/innodb/randomball_normal_1000.0.csv};
\end{axis}
\end{tikzpicture}
\label{fig:perc}
}
	\caption{Deviation of the maximum and 95th percentile weight leaves from uniform as observed in \innodb{}.}\label{fig:deviation}
}
\end{figure}
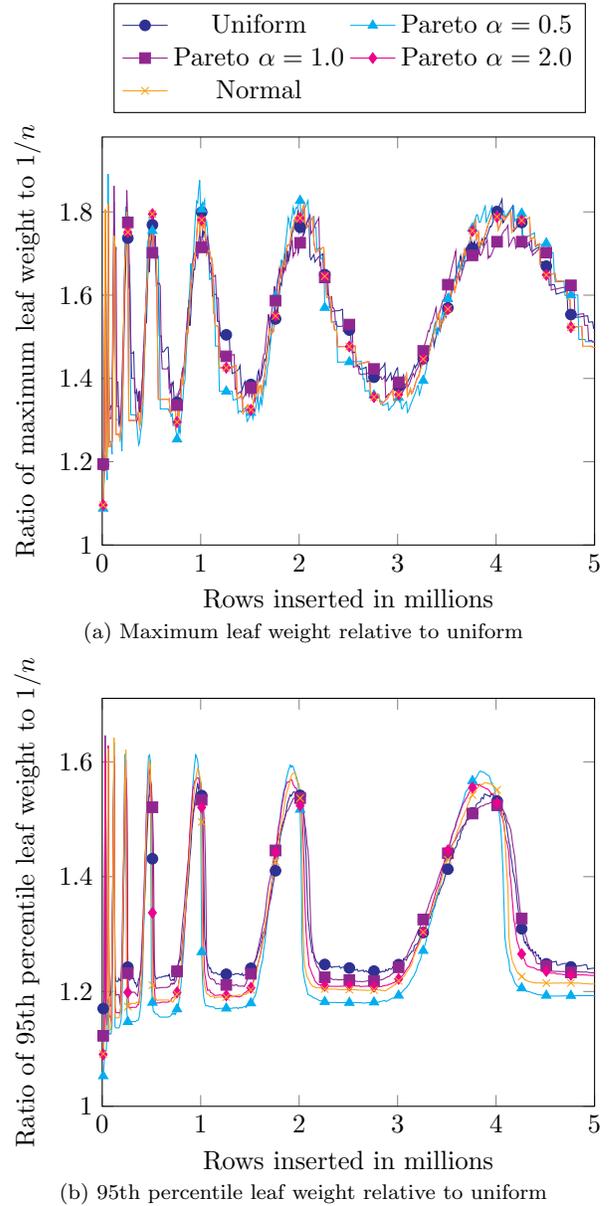


%% file: synthetic-plots.tex

\begin{figure*}[p]
{\centering
	\ref{synth-legend}\\
	\subfloat[Uniform key distribution]{
		\tikzsetnextfilename{synthetic-uniform}
		\begin{tikzpicture}
		\begin{axis}[
				width=0.48\textwidth,
				height=0.31\textheight,
				xlabel=Rows inserted in millions,
				xlabel near ticks,
				ylabel near ticks,
				ylabel style={align=center, text width=0.27\textheight},
				ylabel=Recycle rate (last {50,000} insertions),
				xmin=0.05,
				xmax=5,
				ymin=0,
				ymax=50,
				mark repeat={50},
				restrict y to domain=0:50,
				legend to name=synth-legend,
				legend columns = 3,
				]
				\addplot [mark=*, blue, mark size=2] table [x=insertions, y expr=50000.0/\thisrow{flushes}, col sep=comma] {./data/simulations/gg_uni.csv};
				\addlegendentry{Golden Gate}
				\addplot [mark=square*, red, mark size=2] table [x=insertions, y expr=50000.0/\thisrow{flushes}, col sep=comma] {./data/simulations/rb_uni.csv};
				\addlegendentry{Random Ball}
				\addplot [mark=triangle*, green, mark size=2] table [x=insertions, y expr=50000.0/\thisrow{flushes}, col sep=comma] {./data/simulations/fb_uni.csv};
				\addlegendentry{Fullest Bin}
			\end{axis}
		\end{tikzpicture}
		\label{fig:synth-uni}
	}\hfill
	\subfloat[Pareto-0.5 key distribution]{
		\tikzsetnextfilename{synthetic-pareto-5}
		\begin{tikzpicture}
			\begin{axis}[
				width=0.48\textwidth,
				height=0.31\textheight,
				xlabel=Rows inserted in millions,
				xlabel near ticks,
				ylabel near ticks,
				ylabel style={align=center, text width=0.27\textheight},
				ylabel=Recycle rate (last {50,000} insertions),
				xmin=0.05,
				xmax=5,
				ymin=0,
				ymax=50,
				mark repeat={50},
				restrict y to domain=0:50,
				]
				\addplot [mark=*, blue, mark size=2] table [x=insertions, y expr=50000.0/\thisrow{flushes}, col sep=comma] {./data/simulations/gg_par5.csv};
				\addplot [mark=square*, red, mark size=2] table [x=insertions, y expr=50000.0/\thisrow{flushes}, col sep=comma] {./data/simulations/rb_par5.csv};
				\addplot [mark=triangle*, green, mark size=2] table [x=insertions, y expr=50000.0/\thisrow{flushes}, col sep=comma] {./data/simulations/fb_par5.csv};
			\end{axis}
		\end{tikzpicture}
		\label{fig:synth-par5}
	}\\
	\subfloat[Pareto-1 key distribution]{
		\tikzsetnextfilename{synthetic-pareto-1}
		\begin{tikzpicture}
			\begin{axis}[
				width=0.48\textwidth,
				height=0.31\textheight,
				xlabel=Rows inserted in millions,
				xlabel near ticks,
				ylabel near ticks,
				ylabel style={align=center, text width=0.27\textheight},
				ylabel=Recycle rate (last {50,000} insertions),
				xmin=0.05,
				xmax=5,
				ymin=0,
				ymax=50,
				mark repeat={50},
				restrict y to domain=0:50,
				]
				\addplot [mark=*, blue, mark size=2] table [x=insertions, y expr=50000.0/\thisrow{flushes}, col sep=comma] {./data/simulations/gg_par1.csv};
				\addplot [mark=square*, red, mark size=2] table [x=insertions, y expr=50000.0/\thisrow{flushes}, col sep=comma] {./data/simulations/rb_par1.csv};
				\addplot [mark=triangle*, green, mark size=2] table [x=insertions, y expr=50000.0/\thisrow{flushes}, col sep=comma] {./data/simulations/fb_par1.csv};
			\end{axis}
		\end{tikzpicture}
	}\hfill
	\subfloat[Pareto-2 key distribution]{
		\tikzsetnextfilename{synthetic-pareto-2}
		\begin{tikzpicture}
			\begin{axis}[
				width=0.48\textwidth,
				height=0.31\textheight,
				xlabel=Rows inserted in millions,
				xlabel near ticks,
				ylabel near ticks,
				ylabel style={align=center, text width=0.27\textheight},
				ylabel=Recycle rate (last {50,000} insertions),
				xmin=0.05,
				xmax=5,
				ymin=0,
				ymax=50,
				mark repeat={50},
				restrict y to domain=0:50,
				]
				\addplot [mark=*, blue, mark size=2] table [x=insertions, y expr=50000.0/\thisrow{flushes}, col sep=comma] {./data/simulations/gg_par2.csv};
				\addplot [mark=square*, red, mark size=2] table [x=insertions, y expr=50000.0/\thisrow{flushes}, col sep=comma] {./data/simulations/rb_par2.csv};
				\addplot [mark=triangle*, green, mark size=2] table [x=insertions, y expr=50000.0/\thisrow{flushes}, col sep=comma] {./data/simulations/fb_par2.csv};
			\end{axis}
		\end{tikzpicture}
		\label{fig:synth-par2}
	}\\
	\vspace*{-0.015\textheight}
	\subfloat[Normal key distribution]{
		\tikzsetnextfilename{synthetic-normal}
		\begin{tikzpicture}
			\begin{axis}[
				width=0.48\textwidth,
				height=0.31\textheight,
				xlabel=Rows inserted in millions,
				xlabel near ticks,
				ylabel near ticks,
				ylabel style={align=center, text width=0.27\textheight},
				ylabel=Recycle rate (last {50,000} insertions),
				xmin=0.05,
				xmax=5,
				ymin=0,
				ymax=50,
				mark repeat={50},
				restrict y to domain=0:50,
				]
				\addplot [mark=*, blue, mark size=2] table [x=insertions, y expr=50000.0/\thisrow{flushes}, col sep=comma] {./data/simulations/gg_norm.csv};
				\addplot [mark=square*, red, mark size=2] table [x=insertions, y expr=50000.0/\thisrow{flushes}, col sep=comma] {./data/simulations/rb_norm.csv};
				\addplot [mark=triangle*, green, mark size=2] table [x=insertions, y expr=50000.0/\thisrow{flushes}, col sep=comma] {./data/simulations/fb_norm.csv};
			\end{axis}
		\end{tikzpicture}
		\label{fig:synth-norm}
	}\hfill
	\subfloat[Ratio of golden gate to random ball (normally distributed keys)]{
		\tikzsetnextfilename{synthetic-norm-ratio}
		\pgfplotstableread[col sep = comma]{./data/simulations/gg_norm.csv}\goldengate
		\pgfplotstableread[col sep = comma]{./data/simulations/rb_norm.csv}\randomball
		\pgfplotstablecreatecol[copy column from table={\randomball}{flushes}] {rb_flushes} {\goldengate}
		\begin{tikzpicture}
			\begin{axis}[
				width=0.48\textwidth,
				height=0.31\textheight,
				xlabel=Rows inserted in millions,
				xlabel near ticks,
				ylabel near ticks,
				ylabel style={align=center, text width=0.30\textheight},
				ylabel=Recycle rate ratio (last {5,000} insertions),
				xmin=0.05,
				xmax=5,
				ymin=0.8,
				ymax=2,
				mark repeat={50},
				]
				\addplot [mark=*, Plum, no marks] table [x=insertions, y expr=\thisrow{rb_flushes}/\thisrow{flushes}] \goldengate;
			\end{axis}
		\end{tikzpicture}
		\label{fig:synth-ratio}
	}
	\caption{\label{fig:synth}Simulated results with various key distributions and recycling strategies. (Higher is better)}
}
\end{figure*}
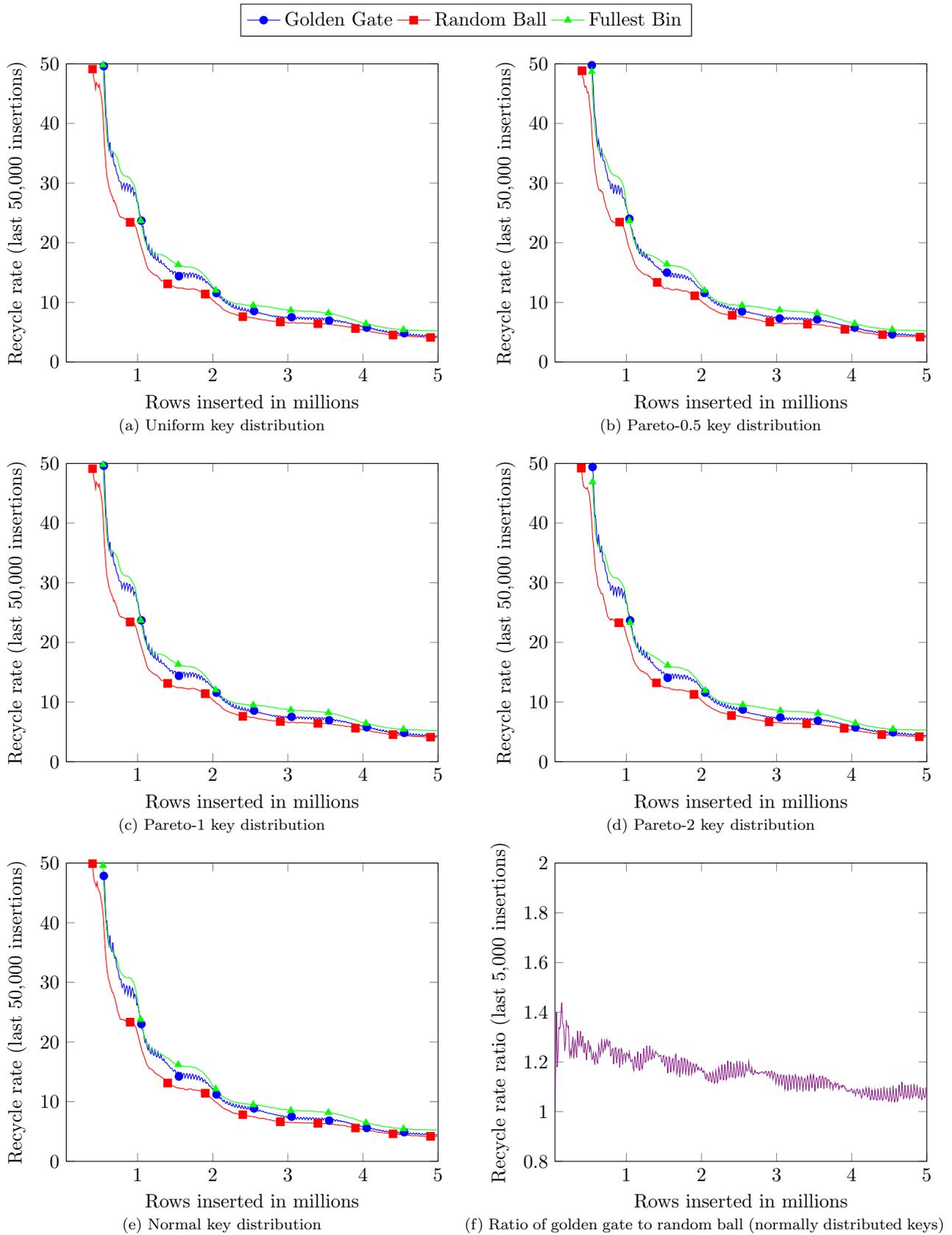


%% file: innodb-plots.tex

\begin{figure*}
{\centering
~\ref{innodb-legend}~\\
\subfloat[Uniform key distribution]{
\tikzsetnextfilename{innodb-uniform}
\begin{tikzpicture}
\begin{axis}[
	width=0.49\textwidth,
	height=0.31\textheight,
	xlabel=Rows inserted in millions,
	xlabel near ticks,
	ylabel near ticks,
	ylabel style={align=center, text width=0.27\textheight},
	ylabel=Recycle rate (last {10,000} insertions),
	xmin=0,
	xmax=5,
	ymin=0,
	ymax=50,
	mark repeat={50},
	legend to name=innodb-legend,
	legend columns=3,
	]
	\addplot [mark=*, blue, mark size=2] table [x=insertions, y=observed, col sep=comma] {./data/innodb/goldengate_uniform_1.0.csv};
	\addlegendentry{Golden Gate}
	\addplot [mark=square*, red, mark size=2] table [x=insertions, y=observed, col sep=comma] {./data/innodb/randomball_uniform_1.0.csv};
	\addlegendentry{Random Ball (default)}
\end{axis}
\end{tikzpicture}
\label{fig:innodbuni}
}\hfill
\subfloat[Pareto $\alpha=0.5$ key distribution]{
\tikzsetnextfilename{innodb-pareto5}
\begin{tikzpicture}
\begin{axis}[
	width=0.49\textwidth,
	height=0.31\textheight,
	xlabel=Rows inserted in millions,
	xlabel near ticks,
	ylabel near ticks,
	ylabel style={align=center, text width=0.27\textheight},
	ylabel=Recycle rate (last {10,000} insertions),
	xmin=0,
	xmax=5,
	ymin=0,
	ymax=50,
	mark repeat={50},
	]
	\addplot [mark=*, blue, mark size=2] table [x=insertions, y=observed, col sep=comma] {./data/innodb/goldengate_pareto_0.5.csv};
	\addplot [mark=square*, red, mark size=2] table [x=insertions, y=observed, col sep=comma] {./data/innodb/randomball_pareto_0.5.csv};
\end{axis}
\end{tikzpicture}
\label{fig:innodbpar0.5}
}\\
\subfloat[Pareto $\alpha=1$ key distribution]{
\tikzsetnextfilename{innodb-pareto1}
\begin{tikzpicture}
\begin{axis}[
	width=0.49\textwidth,
	height=0.31\textheight,
	xlabel=Rows inserted in millions,
	xlabel near ticks,
	ylabel near ticks,
	ylabel style={align=center, text width=0.27\textheight},
	ylabel=Recycle rate (last {10,000} insertions),
	xmin=0,
	xmax=5,
	ymin=0,
	ymax=50,
	mark repeat={50},
	]
	\addplot [mark=*, blue, mark size=2] table [x=insertions, y=observed, col sep=comma] {./data/innodb/goldengate_pareto_1.0.csv};
	\addplot [mark=square*, red, mark size=2] table [x=insertions, y=observed, col sep=comma] {./data/innodb/randomball_pareto_1.0.csv};
\end{axis}
\end{tikzpicture}
\label{fig:innodbpar1}
}\hfill
\subfloat[Pareto $\alpha=2$ key distribution]{
\tikzsetnextfilename{innodb-pareto2}
\begin{tikzpicture}
\begin{axis}[
	width=0.49\textwidth,
	height=0.31\textheight,
	xlabel=Rows inserted in millions,
	xlabel near ticks,
	ylabel near ticks,
	ylabel style={align=center, text width=0.27\textheight},
	ylabel=Recycle rate (last {10,000} insertions),
	xmin=0,
	xmax=5,
	ymin=0,
	ymax=50,
	mark repeat={50},
	]
	\addplot [mark=*, blue, mark size=2] table [x=insertions, y=observed, col sep=comma] {./data/innodb/goldengate_pareto_2.0.csv};
	\addplot [mark=square*, red, mark size=2] table [x=insertions, y=observed, col sep=comma] {./data/innodb/randomball_pareto_2.0.csv};
\end{axis}
\end{tikzpicture}
\label{fig:innodbpar2}
}\\
\subfloat[Normal key distribution]{
\tikzsetnextfilename{innodb-normal}
\begin{tikzpicture}
\begin{axis}[
	width=0.49\textwidth,
	height=0.31\textheight,
	xlabel=Rows inserted in millions,
	xlabel near ticks,
	ylabel near ticks,
	ylabel style={align=center, text width=0.27\textheight},
	ylabel=Recycle rate (last {10,000} insertions),
	xmin=0,
	xmax=5,
	ymin=0,
	ymax=50,
	mark repeat={50},
	]
	\addplot [mark=*, blue, mark size=2] table [x=insertions, y=observed, col sep=comma] {./data/innodb/goldengate_normal_1000.0.csv};
	\addplot [mark=square*, red, mark size=2] table [x=insertions, y=observed, col sep=comma] {./data/innodb/randomball_normal_1000.0.csv};
\end{axis}
\end{tikzpicture}
\label{fig:innodbnorm}
}\hfill
\subfloat[Buffer size and recycle rate]{
\tikzsetnextfilename{innodb-buffer-size}
\begin{tikzpicture}
\begin{axis}[
	width=0.49\textwidth,
	height=0.31\textheight,
	xlabel=Buffer size in MiB,
	xlabel near ticks,
	ylabel near ticks,
	ylabel style={align=center, text width=0.27\textheight},
	ylabel=Recycle rate (last million insertions),
	xmin=0,
	xmax=128,
	ymin=0,
	]
	\addplot [mark=*, Plum, mark size=2] table [x=Buffer, y=Rate, col sep=comma] {./data/innodb_buffer_size.csv};
\end{axis}
\end{tikzpicture}
\label{fig:innodbbuffer}
}

\caption{\innodb{} Insertion buffer recycle rates for various key distributions and memory sizes. (Higher is better)}}
\end{figure*}
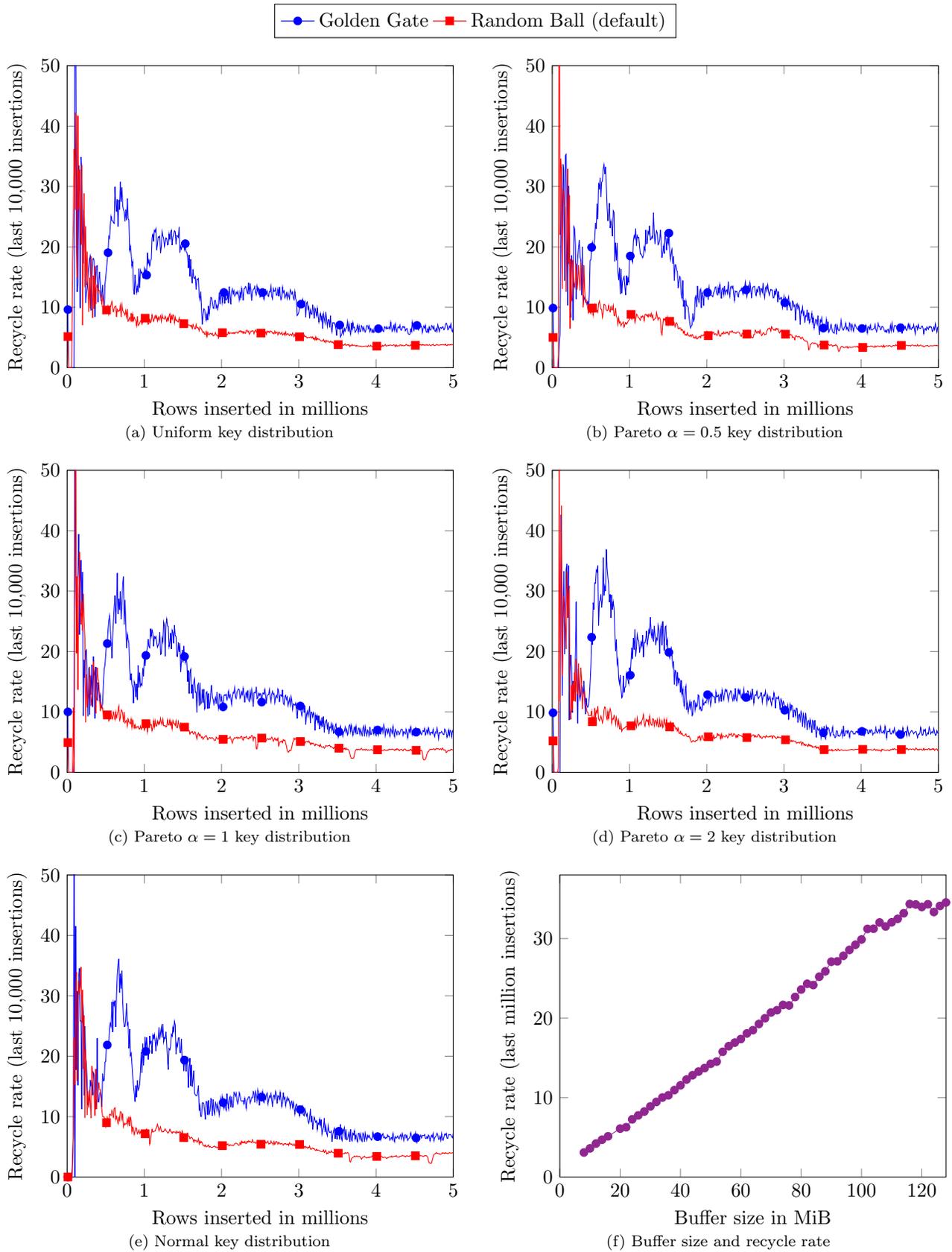


%% file: related.tex

\secput{related}{Related Work}

Balls-and-bins games are one of the most studied models in all of computer
science, so it would be impossible to do justice to the entire literature here.
Rather, we focus on prior work on \defn{dynamic} balls-and-bins games.  In
dynamic balls-and-bins games, balls are added and removed from bins according
to some rules, and the goal is to understand the long-term behavior of the
system.  Thus, ball-recycling games are instances of this general class,
although previously studied dynamic balls-and-bins games are quite different,
and are typically used to study load-balancing problems, rather than
throughput.

Previously studied models differ from ball-recycling games in several
ways: 
\begin{itemize}
\item The process for removing balls is assumed to be random, e.g. a
  random ball is removed, or a random bin is selected for emptying.
  Prior research has assumed these events are determined by some
  external process, so they have not studied algorithms for
  controlling this process.  As a result, although the theory of
  Markov chains has played a major role in the study of balls-and-bins
  games, the theory of Markov decision processes, which we use
  extensively, has not shown up at all.
\item The balls are thrown uniformly randomly.  This is a natural
  assumption for hashing and randomized load-balancing problems, but
  is not appropriate when studying updates to existing keys in a
  database.
\item The objective is to analyze the occupancy of the fullest bin or,
  in some models, to analyze the amount of time that balls wait in a
  queue.  Our objective is to analyze the number of balls recycled in
  each time step.
\item The number of balls in the system is not fixed.  Prior models
  were used for load balancing, in which balls correspond to tasks and
  bins resources, so it makes sense to model new balls entering the
  system asynchronously.  In our setting, the balls are the
  resource (i.e. they correspond to slots in a buffer), so
  they are fixed.
\item The study of dynamic balls-and-bins games was introduced
  simultaneously with the study of the power of multiple
  choices~\cite{AzarBrKa94}, and most past work on dynamic
  balls-and-bins games has been in the same model.  In our model,
  balls have only a single choice (i.e. their on-disk location),
  although it may be possible to extend our work to systems in which
  each ball has multiple choices (e.g. for an insertion buffer for an
  on-disk cuckoo hash table).
\end{itemize}

Azar, et al.~\cite{AzarBrKa94} introduced both the power of two
choices and dynamic balls and bins games.  They showed that if, at
each time step, a random ball is rethrown with $d$ uniformly random
choices, then, in $n^3$ time steps, the fullest bin has $\ln \ln n /
\ln d + O(1)$ balls w.h.p.

In his dissertation~\cite{Mitzenmacher96}, Mitzenmacher studied several dynamic
load-balancing problems.  This included several variants on the supermarket
model, in which balls arrive according to a Poisson process and enqueue
themselves in the shortest of $d$ queues that they select uniformly randomly
from $n$ queues.  Mitzenmacher showed that $d>1$ exponentially reduced the
average time a ball spent in a queue.  He also studied a variant in which, at
each time step, one ball was removed from one queue and was immediately
re-enqueued according to the above procedure.  Adler, et al.~\cite{AdlerBeSc98}
studied a variant of the supermarket model in which balls arrived in batches of
size $m$ and chose their queues in parallel, and showed that the average
waiting time remains $O(\ln \ln n)$ as long as $m$ is sufficiently smaller than
$n$.

Cole, et al.~\cite{ColeFrMa98} studied a model in which balls are recycled
one-at-a-time according to a recycling plan chosen in advance.  In their model,
there are an infinite number of labeled balls, and the adversary specifies in
advance two ball IDs for each time step: the first ID specifies a ball to be
removed from the system, and the second ID specifies a ball to be inserted.
Thus the number of balls currently in the system is always $n$.  The first time
a ball is inserted, it chooses $d$ bins uniformly at random and picks the least
loaded. From then on, whenever that ball reenters the system, it always goes
into that bin.  In their model, the adversary cannot examine the state of the
current system when deciding which ball to recycle, as in our model.  Given
this restriction, they show that the fullest bin has roughly $\ln \ln n / \ln
d$ balls.  V\"{o}cking~\cite{Voecking03} showed the shocking result that, by
choosing bins non-uniformly and breaking ties asymmetrically, the max load
could be reduced to $\ln \ln n / d\ln\phi_d+O(1)$ w.h.p.

Cole, et al.~\cite{ColeMaMe98} extended their results to routing through a
network: an adversary specified in advance the start time, end time, source,
and destination of flows, and the system used a power of two choices variant of
Valiant's randomized routing paradigm~\cite{ValiantBr81} to limit congestion to
$O(\ln \ln n)$ w.h.p.

Czumaj and Stemann~\cite{CzumajSt97} study a load balancing problem in which,
at each time step, a random ball is removed from the system and a new ball is
thrown using $d$ choices.  After the new ball is inserted, the $d$ bins
examined during its insertion are rebalanced.  Surprisingly, the max load is
still $O(\ln \ln n)$.
